\documentclass[a4paper,10pt]{article}
\usepackage[bbgreekl]{mathbbol}
\usepackage[small]{titlesec}
\usepackage[utf8]{inputenc}
\usepackage[T1]{fontenc}  

\usepackage[backend=biber,style=numeric,natbib=false,giveninits=true,doi=true,isbn=false,url=false,date=year,minnames=2, maxnames=3]{biblatex}
\DeclareNameAlias{default}{family-given}

\DeclareFieldFormat[article]{title}{#1}
\renewbibmacro{in:}{}
\DeclareFieldFormat[article]{pages}{#1}
\DeclareFieldFormat{journal}{#1}
\renewcommand\bf\bfseries

\renewbibmacro*{volume+number+eid}{%
	\printfield{volume}%
	\setunit*{\addnbspace}
	\printfield{number}%
	\setunit{\addcomma\space}%
	\printfield{eid}}
\DeclareFieldFormat[article]{number}{\mkbibparens{#1}}
\DeclareFieldFormat[article]{volume}{\textbf{#1}}
\DeclareFieldFormat{year}{\mkbibparens{#1}}
\DeclareBibliographyDriver{article}{%
	\printnames{author}:%
	\newunit\newblock
	\printfield{title}%
	\newunit\newblock
	\printfield{journaltitle}
	\newunit
	\iffieldundef{number}{\printfield{volume}}{\printfield{volume}\addspace\printfield{number}}
	\addcomma\addspace\printfield{pages}\addspace
	\printfield{year}
}
\AtEveryBibitem{\clearfield{month}}
\AtEveryBibitem{\clearfield{day}}
\usepackage[english]{babel}
\usepackage[T1]{fontenc}
\usepackage{csquotes}
\usepackage{bbm}
\usepackage{amsmath,amsfonts,amsthm,amsbsy,amssymb,dsfont,stmaryrd}
\usepackage[dvipsnames]{xcolor}
\usepackage{braket}
\usepackage{tikz}
\usepackage{tikzscale}
\usepackage{tikz-cd}

\usetikzlibrary{arrows}
\usetikzlibrary{patterns}
\usetikzlibrary{calc}
\usetikzlibrary{shapes.geometric}
\usetikzlibrary{hobby}
\usetikzlibrary{shapes.geometric, arrows,chains}

\usepackage{mathtools}

\numberwithin{equation}{section}

\newcommand\myshade{85}
\colorlet{mylinkcolor}{violet}
\colorlet{mycitecolor}{YellowOrange}
\colorlet{myurlcolor}{Aquamarine}

\usepackage[left=2.5cm,right=2.5cm,top=2.5cm,bottom=2.5cm]{geometry}
\usepackage[unicode=true,pdfusetitle,
bookmarks=true,bookmarksnumbered=false,bookmarksopen=false,
breaklinks=false,pdfborder={0 0 1},backref=false,
linkcolor  = mylinkcolor!\myshade!black,
citecolor  = mycitecolor!\myshade!black,
urlcolor   = myurlcolor!\myshade!black,
colorlinks = true,
]
{hyperref}
\usepackage[nameinlink]{cleveref}

\usepackage{tikz}
\usetikzlibrary{arrows}
\usetikzlibrary{intersections}

\usepackage{graphicx}
\usepackage{caption}

\definecolor{ct_black}{HTML}{000000}
\definecolor{ct_orange}{HTML}{ED872D}
\definecolor{ct_purple}{HTML}{7A68A6}
\definecolor{ct_blue}{HTML}{348ABD}
\definecolor{ct_turquoise}{HTML}{188487}
\definecolor{ct_red}{HTML}{E32636}
\definecolor{ct_pink}{HTML}{CF4457}
\definecolor{ct_green}{HTML}{467821}

\definecolor{ct2_green}{HTML}{9FF781}
\definecolor{ct2_green_dark}{HTML}{088A08}

\theoremstyle{plain}
\newtheorem{thm}{\protect\theoremname}[section]
\newtheorem*{thm*}{\protect\theoremname}
\theoremstyle{plain}
\newtheorem{lem}[thm]{\protect\lemmaname}
\theoremstyle{plain}

\theoremstyle{plain}
\newtheorem{prop}[thm]{\protect\propositionname}
\theoremstyle{remark}

\theoremstyle{remark}
\newtheorem{rem}[thm]{\protect\remarkname}
\newtheorem*{rem*}{\protect\remarkname}
\theoremstyle{definition}
\newtheorem{defn}[thm]{\protect\definitionname}
\theoremstyle{plain}

\providecommand{\assumptionname}{Assumption}
\providecommand{\claimname}{Claim}
\providecommand{\corollaryname}{Corollary}
\providecommand{\definitionname}{Definition}
\providecommand{\lemmaname}{Lemma}
\providecommand{\propositionname}{Proposition}
\providecommand{\remarkname}{Remark}
\providecommand{\theoremname}{Theorem}
\providecommand{\examplename}{Example}

\crefname{section}{Section}{Sections}
\crefname{appendix}{Appendix}{Appendices}
\crefname{figure}{Figure}{Figures}
\crefname{assumption}{Assumption}{Assumptions}
\crefname{thm}{Theorem}{Theorems}
\crefname{lem}{Lemma}{Lemmas}
\crefformat{equation}{(#2#1#3)}

\crefrangelabelformat{equation}{(#3#1#4--#5#2#6)}

\crefmultiformat{equation}{(#2#1#3}{, #2#1#3)}{#2#1#3}{#2#1#3}
\Crefmultiformat{equation}{(#2#1#3}{, #2#1#3)}{#2#1#3}{#2#1#3}

\newtheorem*{lem*}{\protect\lemmaname}

\newcommand{\LOC}[1]{\mathrm{LOC}_{#1}}
\newcommand{\ee}{\operatorname{e}}
\newcommand{\ii}{\operatorname{i}}

\newcommand{\ZZ}{\mathbb{Z}}

\newcommand{\NN}{\mathbb{N}}
\newcommand{\RR}{\mathbb{R}}
\newcommand{\CC}{\mathbb{C}}
\newcommand{\PP}{\mathbb{P}}

\newcommand{\WW}{\mathbb{W}}
\newcommand{\bbLambda}{\mathbb{\Lambda}}

\newcommand{\calN}{\mathcal{N}}
\newcommand{\calF}{\mathcal{F}}

\newcommand{\hQ}{\hat{Q}}

\newcommand{\norm}[1]{\left\|#1\right\|}
\newcommand{\ip}[2]{\langle #1, #2 \rangle}

\newcommand{\dif}{\operatorname{d}}
\newcommand{\tr}{\operatorname{tr}}
\renewcommand{\Im}[1]{\operatorname{\mathbb{I}\mathbbm{m}}\{#1\}}
\renewcommand{\Re}[1]{\operatorname{\mathbb{R}\mathbbm{e}}\{#1\}}
\newcommand{\ve}{\varepsilon}
\newcommand{\vf}{\varphi}
\newcommand{\Id}{\mathds{1}}
\newcommand{\HH}{\mathcal{H}}

\newcommand{\BLO}[1]{\mathcal{B}(#1)}

\newcommand{\dist}{\mathrm{dist}}

\renewcommand{\vec}[1]{\mathbf{#1}}

\newcommand{\findex}{\operatorname{index}}

\newcommand{\supp}{\operatorname{supp}}

\newcommand{\im}{\operatorname{im}}
\newcommand{\ad}{\operatorname{ad}}
\newcommand{\Ad}{\operatorname{Ad}}

\newcommand{\HHB}{\mathcal{H}}

\newcommand{\HHE}{\hat{\mathcal{H}}}
\newcommand{\HE}{\hat{H}}

\newcommand{\WLOC}{\mathrm{WLOC}}
\newcommand{\WLOCC}[1]{\mathrm{WLOC}_{#1}}
\newcommand{\qProj}[1]{\mathcal{P}_{#1}}

\usepackage{environ}

\title{Fredholm Homotopies for Strongly-Disordered 2D Insulators}
\author{Alex Bols\\\footnotesize{QMATH, Department of Mathematical Sciences, University of Copenhagen}\footnote{alex-b@math.ku.dk}\\Jeffrey Schenker\\
	\footnotesize{Department of Mathematics, Michigan State University}\footnote{schenke6@msu.edu}\\Jacob Shapiro\\
	{\footnotesize Department of Physics, Princeton University}\footnote{jacobshapiro@princeton.edu}
}

\begin{document}
	
	\maketitle

	\begin{abstract}
		We study topological indices of Fermionic time-reversal invariant topological insulators in two dimensions, in the regime of strong Anderson localization. We devise a method to interpolate between certain Fredholm operators arising in the context of these systems. We use this technique to prove the bulk-edge correspondence for mobility-gapped 2D topological insulators possessing a (Fermionic) time-reversal symmetry (class AII) and provide an alternative route to a theorem by Elgart-Graf-Schenker \cite{EGS_2005} about the bulk-edge correspondence for strongly-disordered integer quantum Hall systems. We furthermore provide a proof of the stability of the $\ZZ_2$ index in the mobility gap regime. These two-dimensional results serve as a model for the study of higher dimensional $\ZZ_2$ indices.
	\end{abstract}
	
	\section{Introduction}
	
	Topological insulators (TIs) \cite{Hasan_Kane_2010} are characterized chiefly by their insulator condition, which is usually formulated as a \emph{gap requirement}, i.e., that the Hamiltonian has an interval of energies without any states \textemdash a spectral gap. However, some properties of these systems can  be explained only if one has, more generally, a \emph{mobility gap}, namely an interval of energies associated to dynamically localized eigenstates, a situation arising under \emph{strong disorder}. For example, the plateaus of the integer quantum Hall effect require this regime \cite{Graf_2007}, and certain physically-appealing interpretations of topological systems appear precisely in the strongly disordered regime, e.g., for Floquet systems \cite{NathanPRL16,Shapiro2019} with orbital magnetization.
	
	Curiously, all of the studies of topological systems in the mobility gap regime
	\cite{EGS_2005,Graf_Shapiro_2018_1D_Chiral_BEC,Shapiro2019,Shapiro2020} up until now relied crucially on local trace formulas for the invariants to prove the bulk-edge correspondence and stability of the indices. However, such trace formulas seem out of reach for systems whose invariants are $\ZZ_2$-valued, such as time-reversal (TR) invariant 2D systems. Indeed, a local trace formula for $\ZZ_2$ invariants has not been defined yet and there are reasons to doubt its existence altogether (since such a formula should imply continuity, but we know the index should be allowed to jump by even integers). For such $\ZZ_2$ systems, a different approach is required to study strongly-disordered mobility gapped systems without a local trace formula. 
	
	One possibility is to use index formulas in place of trace formulas, an idea first introduced by Bellissard and collaborators \cite{Bellissard_1994} to exhibit integrality and continuity of invariants. In \cite{Graf_Shapiro_2018_1D_Chiral_BEC,Fonseca2020} Fredholm index formulas were used to prove the bulk-edge correspondence, but these proofs were restricted to the spectral gap regime, since it was not clear how to perform homotopies of operators associated with Hamiltonians not possessing a spectral gap. Here we finally extend the Fredholm perspective to the mobility gap regime (by
	\emph{adapting} the scheme of \cite{Fonseca2020}), thus enabling homotopy arguments for strongly-disordered $\ZZ_2$ systems.
	
	Let us briefly describe one of the mathematical problems we shall face. Let $\HH$ be a separable Hilbert space and $\BLO{\HH}$ be the $C^*$ algebra of bounded linear operators on it. If $Q$ is an orthogonal projection and $F$ is Fredholm such that $[Q,F]$ is compact, then $$ \mathbb{Q}F := QFQ+Q^\perp $$ is a Fredholm operator. Indeed, by Atkinson's theorem \cite{Booss_Topology_and_Analysis}, $F$ has a parametrix $G$, and one verifies that $\mathbb{Q}G$ is a parametrix of $\mathbb{Q}F$. Moreover, if $A$ is an orthogonal projection, then \begin{align}
		\mathbb{Q} \exp(2\pi\ii A) \text{ is Fredholm} \label{eq:exponential of projection} \ ;
	\end{align} 
	indeed, $\exp(2\pi\ii A)=\Id$ so $ \mathbb{Q} \exp(2\pi\ii A)$ is the identity operator (thus trivially Fredholm). 
	In the context of disordered 2D-TIs,
	the topological invariants arise as indices of Fredholm operators of  the form $ \mathbb{Q} \exp(2\pi\ii A)$, with $A$ \emph{not a projection}, but nonetheless close to a projection in an approximate sense (see \cref{def:quasi-proj} below) so that \cref{eq:exponential of projection} still holds. This approach was used in \cite{Fonseca2020} to study disordered but still spectrally gapped systems. To deal with mobility gapped systems, we must \emph{further} weaken the sense in which $A$ is an approximate projection, while still keeping the basic logic leading to \cref{eq:exponential of projection}.
	
	An additional major feature of the present paper is the proof that the $\ZZ_2$ index is (deterministically) invariant with respect to the choice of the Fermi energy within the mobility gap. This fact has profound implications for the entire classification theory of topological insulators in the strong disorder. Indeed, our method lends itself to study the strong disorder regime of higher dimensional and different symmetry class systems, although here we merely concentrate on the two dimensional case.
	
	In the physics literature the question of the stability of the $\ZZ_2$ index has also been studied; see e.g. analytically in \cite{PhysRevB.73.045322,Mani_2016}, numerically in \cite{Orth2016,PhysRevMaterials.1.024201,doi:10.1143/JPSJ.80.053703} and experimentally in \cite{doi:10.1126/science.1148047,PhysRevLett.114.096802,PhysRevB.90.115305}. 
	
	
	This paper is organized as follows. We begin by describing the mathematical setting and the sense in which we take an operator to be \emph{approximately a projection.} We go on to define the induced edge system, the TR symmetry operator, and the relevant indices for TR invariant systems. In the heart of the paper, we prove the invariance of the $\ZZ_2$ index in \cref{thm:stability of index} and the bulk edge correspondence in \cref{thm:bec} and finally we delegate some technical points of the proofs to the appendix. Some further discussion about future directions is presented right before the end. Of independent interest may be the appendix discussing the so-called "SULE" basis and its applications, \cref{sec:The SULE basis}.
	
	\section{Setting}
	We consider tight-binding, single-electron models in two-dimensions so that our Hilbert space is either $\HH = \ell^2(\ZZ^2)\otimes\CC^N$ for a bulk, infinite, sample or $\hat{\HH} = \ell^2(\ZZ\times\NN)\otimes\CC^N$ for an edge, half-infinite, sample. Here $N\in\NN_{\geq1}$ is some fixed internal number of degrees of freedom we allow on each site (number of atoms in the unit cell, spin, iso-spin, etc.). Throughout the discussion edge objects shall carry a hat. 
	
	Between these two Hilbert spaces one has the natural injection $$ \iota:\hat{\HH}\hookrightarrow\HH $$ which extends a half-space function $\phi\in \hat{\HH}$
	into the full plane by taking $\iota\phi$ to be
	zero in the lower half plane. With its adjoint $\iota^\ast$ (restriction to the half-space) we find the relations $$ \iota^\ast\iota = \Id_{\hat{\HH}}\,;\quad \iota\iota^\ast = \Lambda_2 \ , $$ with $\Lambda_2$ being  orthogonal projection onto the upper half plane. We also will work with the projection onto the RHS of space, $\Lambda_1$, and oftentimes use the notation 
	\begin{equation}\label{eq:partialj}
		\partial_j A \equiv -\ii [\Lambda_j,A]
	\end{equation}for the non-commutative derivative of $A$ in direction $j=1,2$. 
	
	\subsection{Spatial constraints} If $A$ is an operator on either Hilbert space and $x,y$ are  points in $\ZZ^2$ (or $\ZZ\times\NN$) then $A_{xy}\equiv\ip{\delta_x}{A \delta_y}$ is an $N\times N$ matrix with entries in $\CC$. We are interested in formulating decay  in terms of its matrix norm as $\norm{x-y}\to\infty$ (off-diagonal decay) or as $x,y\to\infty$ (diagonal decay). The following notions are recalled from \cite[Section 3]{Shapiro2019}. They are formulated for $\HH$ but have analogous definitions for operators on $\hat{\HH}$.
	
	\begin{defn}[Local operator]\label{def:local operator}
		An operator $A\in\BLO{\HH}$ is \emph{local} if and only if for each $\alpha \in \NN$
		sufficiently large, 
		there exists $C_\alpha <\infty$ such that $$ \norm{A_{xy}} \leq C_\alpha (1+\norm{x-y})^{-\alpha}\qquad(x,y\in\ZZ^2)\,. $$
		We denote the space of local operators as $\LOC{}$.	
	\end{defn}
	
	Operators in $\LOC{}$ arise when applying the smooth functional calculus to operators which are local with exponential off-diagonal decay rate, see \cite[Lemma A.3]{Elbau_Graf_2002} e.g.. For us these will be smooth functions of Hamiltonians.
	We also have operators whose off-diagonal decay rate is not uniform in the diagonal direction,
	\begin{defn}[Weakly-local operator] An operator $A\in\BLO{\HH}$ is \emph{weakly-local} if and only if there is some $\nu\in\NN$ such that for any $\mu\in\NN$ sufficiently large there is a constant $C_\mu<\infty$ so that 
		\begin{align}
			\norm{A_{xy}} \leq C_\mu (1+\norm{x-y})^{-\mu}(1+\norm{x})^{+\nu}\qquad(x,y\in\ZZ^2)\,. 	\label{eq:weakly-local operator}
		\end{align}
		We denote the space of all weakly-local operators as $\WLOC$.
	\end{defn}
	
	In our application, the sufficiently large threshold for $\mu$ is, say, $10$ and fixed throughout the paper (it is dictated by the finite number of algebraic operations in the $\WLOC$ *-algebra, see \cite[Section 3]{Shapiro2019}).
	
	Of course if $\nu = 0$, we recover the notion of a local operator. Such weakly-local
	estimates arise (almost surely) for measurable functions of random Hamiltonians which exhibit Anderson localization \cite{EGS_2005}, e.g., $\chi_{(-\infty,\mu)}(H)\in\WLOC$ almost-surely for a Hamiltonian $H$ which is Anderson localized around $\mu$. In \cite[Section 3]{Shapiro2019} it was shown that $\WLOC$ is a *-algebra.
	
	Finally, we have the notion of confined operators,
	with matrix elements that decay along one axis:
	\begin{defn}[Weakly-local and confined operator] An operator $A\in\BLO{\HH}$ is \emph{weakly-local and confined in direction $j$ ($j=1$ or $2$)} if and only if there is some $\nu>0$ such that for any $\mu>0$ sufficiently large there is a constant $C_\mu<\infty$ so that 
		\begin{align}
			\norm{A_{xy}} \leq C_\mu(1+\norm{x-y})^{-\mu}(1+|x_j|)^{-\mu}(1+\norm{x})^{+\nu}\qquad(x,y\in\ZZ^2)\,. \label{eq:weakly-local and confined estimate}
		\end{align}
		We denote the space of such operators as $\WLOCC{j}$. The space of operators obeying the estimate \cref{eq:weakly-local and confined estimate} with $\nu=0$ is denoted $\LOC{j}$.
	\end{defn}
	
	It is a fact that $\WLOCC{j}$ forms a *-closed two-sided ideal within $\WLOC$, see \cite[Section 3]{Shapiro2019}. Moreover, \begin{align}\WLOCC{1}\WLOCC{2} \subseteq \WLOCC{1}\cap\WLOCC{2} \subseteq \mathcal{J}_1 \ , \label{eq:confined operators are trace class}\end{align} the latter space being the trace-class operators. Here the juxtaposition  $\WLOCC{1}\WLOCC{2}$ on the left indicates the set of all pairwise products where each factor comes from the corresponding space. Finally, we have the implication \cite[Cor. 3.16]{Shapiro2019} \begin{align}A\in\WLOC\Longrightarrow\partial_j A \in \WLOCC{j}\, ,\label{eq:derivative of weakly local operator is confined}\end{align}
	with $\partial_jA$ as in \cref{eq:partialj}.
	\subsection{Deterministic mobility-gapped insulators}
	\begin{defn}[Physical system]\label{def:Hamiltonians}
		A \emph{Hamiltonian} is an operator $H\in\LOC{}$ which is self-adjoint.	\end{defn} 
	
	\begin{defn}[Insulator]\label{def:deterministic insulator} Let $\Delta \subset \RR$ be an interval and let $B_1(\Delta)$ denote the set of Borel measurable functions $f$ which are constant on $(-\infty,\inf \Delta]$ and on $[\sup \Delta,\infty)$ with $\norm{f}_{\infty}\le 1$. A Hamiltonian $H$ is \emph{mobility-gapped on}
		$\Delta$ if and only if (1) $f(H)\in\WLOC$ for all $f\in B_1(\Delta)$, where the estimate \cref{eq:weakly-local operator} is uniform in the choice of such $f$, and (2) all eigenvalues of $H$ within $\Delta$ are of finite multiplicity.
	\end{defn}
	Of course, if $H$ is spectrally gapped on $\Delta$ (i.e., $\sigma(H)\cap\Delta=\varnothing$), then  $H$ is mobility-gapped on $\Delta$.  Indeed, in this case $f(H)\in \LOC{}$ (not just $\WLOC$) for $f\in B_1(\Delta)$, with uniform bounds, as one can show using the Combes-Thomas estimate (see \cite[Theorem 10.5]{AizenmanWarzel2016} in our context of discrete Schr\"odinger operators) for the smooth functional calculus \cite{Elbau_Graf_2002} (with regards to $f(H)$, the measurable function $f\in B_1(\Delta)$ maybe be deformed to be smooth as the deformation only affects $\left.f\right|_\Delta$). However, $H$ may have a mobility gap in $\Delta$ while $\sigma(H)\cap\Delta\neq \varnothing$. For example, this property holds (almost surely)
	if $H$ is a random operator exhibiting Anderson localization within $\Delta$. To see this, one combines \cite[Eq-n (E.6)]{Aizenman_Graf_1998} together with \cite[Proposition A.1]{Shapiro2019}. It can be shown that if $\Delta$ is a mobility gap for $H$, then the DC conductivity of $H$, computed using the Kubo formula, vanishes for energies $E\in \Delta$ \cite{Aizenman_Graf_1998}.
	\begin{rem}
	Whether Anderson localization in the strong coupling regime can be proven for a topologically-non-trivial time-reversal-invariant $\ZZ_2$ model remains currently open (cf. \cite{Graf_Shapiro_2018_1D_Chiral_BEC} and the localization proof in \cite{Shapiro_2021} for the chiral topological case or \cite{EGS_2005} for a proof of Anderson localization for the non-trivial IQHE case). See the concrete model we propose to study in the future in \cref{sec:example}.
	\end{rem}
	\subsection{Quasi-projections} We
	now turn to the heart of the matter concerning index theory. We want to define a space of operators which are almost projections as one probes them far away along some axis. 
	\begin{defn}[$j$-quasi-projection]\label{def:quasi-proj} An operator $A\in\BLO{\HH}$ is a \emph{$j$-quasi-projection} ($j=1$ or $2$) if and only if (1) $A$ is self-adjoint, (2) $A\in\WLOC$, and (3) $A^2-A\in\WLOCC{j}$.  
	\end{defn}
	
	We denote the space of such operators as $\qProj{j}$, or $\qProj{j}(\WLOC)$ to emphasize the presence of $\WLOC$ in (1) and (2).  We write $\qProj{j}(\LOC{})$ for the set of $j$-quasi-projections such that $A\in \LOC{}$ and $A^2-A\in \LOC{j}$. Of course orthogonal projections which happen to be in $\LOC{}$ or $\WLOC$ are trivially in these spaces.
	
	What concerns us in this paper is the space $\qProj{j}(\WLOC)$ since our operators arise in the context of Anderson localization, whereas in \cite{Fonseca2020} the relevant object was $\qProj{j}(\mathrm{LOC})$ thanks to the spectral gap condition. In \cite[Prop. A4]{Fonseca2020} it was shown that if $A\in\qProj{2}(\mathrm{LOC})$ then $\exp(2\pi\ii A)-\Id\in\mathrm{LOC}_2$ (cf. \cref{eq:exponential of projection}) so that $[\Lambda_1, \exp(2\pi\ii A)]\in\mathrm{LOC}_1\cap\mathrm{LOC}_2$. In particular, thanks to \cref{eq:confined operators are trace class} the commutator is trace-class, hence compact, so that for such $A$, $\mathbb{\Lambda}_1 \exp(2\pi\ii A) \ = \ \Lambda_1 \exp(2\pi \ii A) \Lambda_1 + (1-\Lambda_1)$ is Fredholm. 
	Indeed,  $\exp(2\pi\ii A)-\Id = h(A)(A^2-A)$ with $h$ an analytic function.  Because $A$ is local, $h(A)$ is also local (using the Combes-Thomas estimate). Hence one may use the ideal property of local and confined operators to conclude that $\exp(2\pi\ii A)-\Id\in\mathrm{LOC}_2$.

	On the other hand, in the mobility gap regime, 
	one encounters operators $A$ that are merely weakly-local, but not local.
	For such operators, it is not known (to us) whether a Combes-Thomas estimate holds. Hence it is unclear whether $h(A)$ is also weakly-local. That is, it is an open question whether the weakly-local property is preserved under analytic functions. Hence, we see no
	direct way to prove that $\mathbb{\Lambda}_1 \exp(2\pi\ii A)$ is Fredholm. \\

	\noindent \emph{Open Questions:} Does the Combes-Thomas estimate hold for self-adjoint weakly-local operators?  Is $\WLOC$ closed under convergent power series? \\
	
	We shall deal with this problem by circumventing it, in \cref{thm:Fredholm homotopies in the mobility gap regime}.  Thanks to the fact that we are dealing with Fredholm operators, we can avoid talking about holomorphic functions and instead use  polynomials. Then the algebraic (rather than unestablished topological) properties of $\WLOC$ will suffice.
	
	\subsection{Edge systems}
	\begin{figure}
	    \centering
	    	\begin{center}
		\begin{tikzpicture}[scale=0.8]
			
			\draw [lightgray, fill=lightgray] (-4, 0) rectangle (4,3);

			\draw[very thick,line width=0.05cm,->] (-4,-0) -- (4,-0);
			
			\draw[pattern=north west lines, pattern color=black] (-4,-0.2) rectangle (4,0);
			\draw[very thick,line width=0.05cm,->] (0,0) -- (0,3);

			\node [below left] at (4,0.9) {$\mathbb{Z}$};   
			\node [right] at (0.2,2.8) {$\mathbb{N}$};     
			
			\node [above] at (2, -1.5) {Vacuum};    
			\node [above] at (2, 1) {Material};     
			\node [above] at (0, -1) {Edge};
			
		\end{tikzpicture}
	\end{center} 
	    \caption{The geometry of truncating the system to the edge.}
	    \label{fig:edge truncation}
	\end{figure}
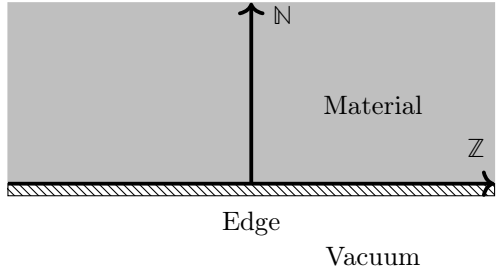
	Edge systems are Hamiltonians on $\HHE$, as in \cref{def:Hamiltonians}. However, unlike bulk systems, they are generically  \emph{not} insulators. Before we turn to the mobility gap regime, it is instructive to consider the spectral gap case; then, there is a convenient way to encode the fact an edge Hamiltonian is associated with a bulk one which possesses a spectral gap, \emph{without making reference to said bulk Hamiltonian}. The following definition is adapted from \cite[Definition 3.4]{shapiro2020tightbinding}:
	
	\begin{defn}[Edge systems with a bulk spectral gap]\label{def:edge system with bulk gap}
		Let $\Delta\subseteq\RR$ be a given interval. An edge system $\hat{H}\in\BLO{\hat{\HH}}$ has a \emph{bulk-spectral-gap} within $\Delta$ if and only if there exists a smooth $g:\RR\to[0,1]$ with $g(\lambda)=1$ for $\lambda<\Delta$,  $g(\lambda)=0$ for $\lambda>\Delta$, and such that $ g(\HE)\in\qProj{2}(\mathrm{LOC})$.
	\end{defn}
	
	The rationale behind this definition is as follows: $g$ differs from $\chi_{(-\infty,\mu)}$, a projection, only within the bulk gap $\Delta$. Edge states with energies in $\Delta$, however, are localized near the edge (see \cref{fig:edge truncation}) and hence become less relevant as one goes into the bulk, whence $$g(\HE)\to g(H)=\chi_{(-\infty,\mu)}(H)\qquad(\text{into the bulk})\,.$$
	
	The more intuitive notion of an edge system (which calls for a given bulk insulator) remains appropriate in the mobility gap regime:
	\begin{defn}[Edge systems with a bulk mobility gap]\label{def:edge system with bulk mobility gap}
		An edge Hamiltonian $\HE\in\LOC{}(\HHE)$ has a \emph{bulk-mobility-gap} within $\Delta$ if and only if
		\begin{align} \HE-\Ad_\iota K\in\LOC{2}\,  \label{eq:edge Hamiltonian compatible with a bulk one}\end{align}
		for some bulk insulator $K\in \LOC{}(\HH)$ within  $\Delta$ (as in \cref{def:deterministic insulator}). Here we use the
		notation $\Ad_\iota K \equiv \iota^\ast K\iota$.
	\end{defn}
	
	In \cite[Lemma A3 iii]{Elbau_Graf_2002} one can find a proof that \cref{def:edge system with bulk mobility gap} implies \cref{def:edge system with bulk gap} if the bulk Hamilotnian is spectrally gapped.

	\subsection{Time-reversal symmetry and the $\Theta$-odd Fredholm index}
	Time reversal (TR) symmetry is a $\CC$-conjugate-linear and anti-unitary map $\Theta:\HH\to\HH$, i.e., $$ \Theta (\psi +\alpha \phi) = \Theta\psi +\bar{\alpha} \Theta \phi \quad \text{and} \quad \langle\Theta\psi,\Theta\vf\rangle = \langle\vf,\psi\rangle\qquad(\psi,\vf\in\HH, \ \alpha \in \CC) \ , $$ 
	such that \begin{align}[\Theta,X_j]=0\,,\label{eq:TR commutes with position operators}\end{align} with $X_j$ being the position operator in direction $j=1,2$. We focus here on the Fermionic case, so we require
	\begin{align}\Theta^2=-\Id\,.\end{align}
	
	We shall assume a TR symmetry $\Theta$ is defined on the
	bulk Hilbert space.  Due to \cref{eq:TR commutes with position operators}, $\iota^*\Theta \iota$ is a TR-symmetry on the edge Hilbert space.  Since the two TR operators are equal fiberwise (on fibers over the upper half plane),  we do not distinguish them with different notation. 
	
	A Hamiltonian $H$ is called TR invariant (TRI) if and only if $[H,\Theta]=0$. Thanks to \cref{eq:TR commutes with position operators}, if $H$ is TRI, then the edge Hamiltonian $\Ad_\iota H$ is also TRI.  However, a general edge Hamiltonian $\HE$ that descends from $H$ (according to \cref{def:edge system with bulk mobility gap}) may not be TRI, since the boundary term $\HE-\Ad_\iota H \in \LOC{2}$ need not be so.

	Recall that the space of Fredholm operators, denoted $\calF$, is the set of all bounded linear operators $F$ such that both $\ker F$ and  $\ker F^\ast$ are finite dimensional and $\im F$ is closed (see, e.g., \cite{Booss_Topology_and_Analysis}). For such an operator, the $\ZZ$-valued Fredholm index is $$ \findex F \equiv \dim \ker F - \dim \ker F^\ast , $$ and the resulting map $\findex:\calF\rightarrow \ZZ$ is continuous in the operator norm topology and stable under compact perturbations. In fact, $\findex$ ascends to a bijection between the set of path-connected components of $\calF$ and $\ZZ$.  Moreover, it obeys a logarithmic law $$\findex AB = \findex A + \findex B\,.$$
	
	\begin{defn}[$\Theta$-odd Fredholm operator] A Fredholm operator $F\in\calF$ is called $\Theta$-odd if $$ F = -\Theta F^\ast \Theta\,. $$ The space of all such operators is denoted as $\calF_\Theta$. Clearly $\calF_\Theta\subseteq\calF_0$, the zero-index path-component of $\calF$, thanks to the logarithmic law of $\findex$.
	\end{defn}
	
	Atiyah and Singer introduced a $\ZZ_2$-valued index for $\Theta$-odd operators,	
	$$ \findex_2 F := \dim \ker F \mod 2 \ , $$
	and proved that $\findex_2$ is continuous, stable under compact perturbations, and ascends to a bijection from the set of path-connected components of $\calF_\Theta$ to $\ZZ_2$ \cite{Atiyah1969}.  Different proofs of these facts can be found also in \cite{SB_2015,Fonseca2020}.  (Atiyah and Singer referred to such operators as skew-adjoint and used a slightly different notation.)

	\section{Example: a disordered BHZ model}\label{sec:example}
	The Bernevig, Hughes, and Zhang model (BHZ henceforth) \cite{doi:10.1126/science.1133734} is a model on $\ell^2(\ZZ^2)\otimes\CC^2\otimes\CC^2$ given by $$ H_{\mathrm{BHZ}} =\left(a\Id+\Re{R_1}+\Re{R_2}\right)\otimes\sigma_3\otimes\Id+\Im{R_2}\otimes\sigma_2\otimes\Id+\Im{R_1}\otimes\sigma_1\otimes\sigma_3 $$ where $\sigma_1,\sigma_2,\sigma_3$ are the three Pauli spin matrices, $R_j$ is the right-shift operator along the $j=1,2$ direction: $$ (R_j \psi)(x) \equiv \psi(x-e_j)\qquad(x\in\ZZ^2) $$ $e_j$ the standard basis vectors, and $\Re{A} := \frac{1}{2}(A+A^\ast);\,\Im{A}:=\frac{1}{2\ii}(A-A^\ast)$. Here $a\in\RR$ is a parameter of the model which controls the selection of the topological phase. Time-reversal for this model is defined as $$ \Theta := \mathcal{C} \Id\otimes\Id\otimes(-\ii\sigma_2) $$ and $\mathcal{C}$ is complex-conjugation of scalars on the complex Hilbert space, which makes $\Theta$ anti-$\CC$-linear.
	
	One may verify now that: $\Theta^2=-\Id$ and that $H_{\mathrm{BHZ}}$ has a spectral gap about zero energy for $a\neq0,\pm2$ (via the Bloch decomposition). Later on we will define a topological $\ZZ_2$ index applicable to this model. It is well known \cite{doi:10.1126/science.1133734} that this index is non-trivial for $0<|a|<2$ and trivial for $|a|>2$. 
	
	There are various ways to make this model disordered. One obvious example is:
	\begin{align}
	    H_{\mathrm{BHZ}}^{\mathrm{disordered}} = H_{\mathrm{BHZ}} + \omega(X)\otimes\sigma_3\otimes\Id
	\end{align} where $X$ is the position operator and $\Set{\omega(x)}_{x\in\ZZ^2}$ is a sequence of i.i.d. random real variables, all distributed with some probability measure $\mu$ on $\RR$. Hence the parameter $a$ has been replaced with $\omega$ which varies throughout the plane. We conjecture that depending on the properties of $\mu$, e.g., its support, the topological phase would be selected, and $\mu$ may be chosen so that there is a well-defined topological phase \emph{yet} there is no spectral gap about zero. We postpone the study of this disordered model to future work.
	
	\section{The topological indices and our main result}
	We (mostly) do not distinguish between the IQHE ($\ZZ$-index) and the TRI ($\ZZ_2$-index) cases.  Hence $\findex A $ will either indicate  $\dim\ker A-\dim\ker A^\ast$ or $\dim\ker A\mod 2$ depending on context. Sometimes we will make this explicit by writing $\findex_{(2)}$ which means an equation either holds for $\findex$ or for $\findex_{(2)}$. When $H$ is TRI, the $\ZZ$-index is trivial and it is appropriate to use the $\ZZ_2$ index instead. We will mostly use the fact the index is stable under norm-continuous and compact perturbations, which is true for both $\findex$ and $\findex_2$. 
	
	\subsection{The bulk topological index}
	
	Let $H$ be a $\Delta$-insulator as in \cref{def:deterministic insulator}. Its associated bulk index is given by \begin{align}\calN_{(2)}(H,\Delta) \ := \ \findex_{(2)} \PP U \ \equiv \ \findex_{(2)} (PUP + P^\perp) \ ,  \label{eq:flux insertion index}\end{align} where $P := \chi_{(-\infty,\mu)}(H)$ is the Fermi projection for $H$ at Fermi energy $\mu\in\Delta$, and 
	$$U :=\exp(\ii\arg(X_1+\ii X_2))$$ is the flux insertion operator.  Since $\mu$ falls in the mobility gap $\Delta$ of $H$, so that $P\in \WLOC$, it follows that $\PP U$ is Fredholm and the index is well defined. In \cref{thm:stability of index} below we see that in agreement with our notation, $\calN_{(2)}(H,\Delta)$ indeed does not depend on the choice of $\mu\in\Delta$. In fact, it is clear that $\calN_{(2)}(H,\Delta)$ depends only on the connected component of $\Delta$ (connectedness meant in the sense of \cref{def:deterministic insulator} holding).   
	
	$\calN$ was associated with the Hall conductivity in the integer quantum Hall effect in \cite{Bellissard_1994} (who already noted that it remains well defined in the mobility gap) via the Kubo formula; see also \cite{ASS1994_charge_def}. $\calN_2$ was associated with the Fu-Kane-Mele $\ZZ_2$ invariant by Schulz-Baldes \cite{SB_2015}.
	In \cite{Aizenman_Graf_1998} there is a self contained proof that \begin{align}P\in\WLOC \Longrightarrow \mathbb{P}U\in\calF\,.\end{align}
	
	\subsection{Stability of the bulk index with respect to the Fermi energy}
	In the present section, we make the dependence on the Fermi energy explicit by denoting $P_\mu := \chi_{(-\infty,\mu)}(H)$. We shall prove \begin{thm}\label{thm:stability of index}
		If $\Delta\subseteq\RR$ is a mobility gap for $H$ as in \cref{def:deterministic insulator} and $P_\mu \equiv \chi_{(-\infty,\mu)}(H)$ then the function $$\Delta\ni\mu\mapsto \findex_{(2)} \PP_\mu U \in \ZZ_{(2)}$$ is constant.
	\end{thm}
	
	\begin{rem}This theorem is significant in regards to localization theory because it allows to exhibit explicit time-reversal invariant models for which complete localization in two dimensions fails. Indeed, take any random model with a non-trivial $\ZZ_2$ invariant (e.g. a random version of the BHZ model \cite{BHZ_2006,Fu_Kane_2007}). Then at $\mu$ below the bottom of the spectrum the invariant is clearly zero, whereas (by assumption) somewhere it is non-zero. Hence in going from these two points, there must have been a value of $\mu$ for which $P_\mu$ failed to be $\WLOC$ (for otherwise the index would have been well-defined throughout the interpolation, but we have shown that if that were the case the index would also be constant).  This argument is analogous to one used to prove delocalization for random Landau Hamiltonians via the quantum Hall conducatance \cite{Germinet_Klein_Schenker_2007}, with the difference that in the current setting the operators are TRI, although with $\Theta^2=-1$.  (For TRI systems with $\Theta^2=1$, the scaling theory of localization \cite{Anderson_1979_PhysRevLett.42.673} predicts complete localization, i.e., that $P_\mu\in \WLOC$ for all $\mu$, for systems with extensive randomness; see \cite{Evers_Mirlin_2008} for a review of what is to be expected.)
	\end{rem}
	
	For the IQHE, \cref{thm:stability of index} was proved in \cite[Prop. 3]{EGS_2005} using the Kubo trace formula, and not using the integrality of the index. To the best of our knowledge, the result is new for the $\ZZ_2$ index case. In \cite{Shapiro2020} this statement was boosted to invariance of $\calN$ w.r.t. $H$, however, the argument there does not carry over for $\calN_2$.
	
	We are aware of similar statements that have been made in \cite[Theorem 5]{Schulz-Baldes_2015_Z2} and in \cite[(2.2)]{Katsura_Koma_2018_doi:10.1063/1.5026964}. Howwever, it appears these proofs contain gaps. Indeed, in the former one, it is not clear that $g(H)P=g(H)$ may be arranged in the mobility gap regime, although this identity is used in the proof. In the latter, it is not clear that their equation (2.1) can hold at a fixed value of the Fermi energy $E_F$ for their $H'$ \textemdash \ see our \cref{lem:SULE implies Green function decay} below in which a similar bound is obtained for almost every energy $E$.
	
	\begin{proof}[Proof of \cref{thm:stability of index}]
		Let $\mu'$ be a point nearby $\mu$ within $\Delta$. We want to show that $\findex \PP_\mu U = \findex \PP_{\mu'} U $. Without loss of generality, we take $\mu'>\mu$ and write $P := P_\mu$ and $Q := P_{\mu'}-P$. So $Q$ projects onto the (localized) states within $[\mu,\mu')$, both $P$ and $Q$ are $\mathrm{WLOC}$, and $PQ = 0$. Then $$\PP_{\mu'} U  \ = \  (P+Q)U(P+Q)+(P+Q)^\perp \ = \  PUP+QUQ+(P+Q)^\perp + PUQ +QUP 
		\,. $$ We claim that $PUQ$ is compact.  Indeed, we have
		$$|PUQ|^2  \ = \ (PUQ)^\ast PUQ \ = \ QU^\ast P U Q \ = Q(U^\ast P U-P) Q  \, .$$
		
		Since $U^\ast P U-P$ is Schatten-3 \cite{ASS1994_charge_def}, it follows that $PUQ$ is Schatten-6 and hence compact. As a result,  $$ \findex \PP_{\mu'}U \ = \ \findex \left [ PUP+QUQ+(P+Q)^\perp \right ]\,. $$ Since the three terms within the index on the right hand side are in fact operators on direct sumands of the Hilbert space, we claim that
		$$ \findex \left [ PUP+QUQ+(P+Q)^\perp \right ] \ = \ \findex \left [ \PP U \oplus \mathbb{Q} U \right ] \ .$$
		To see this, note that
		$$\ker \left [ PUP+QUQ+(P+Q)^\perp \right ] \ \cong \ \ker PUP \oplus \ker QUQ \ \cong \ \ker \PP U \oplus \ker \mathbb{Q} U \, ,$$
		and similarly for the adjoint operator (note that  $\ker (P+Q)^\perp =\operatorname{ran} (P+Q)$).
		For either the $\ZZ$ or the $\ZZ_2$ index, we have additivity under direct sums:  $\findex A\oplus B = \findex A +\findex B$ (with addition modulo $2$ in the $\ZZ_2$ case). Thus $\findex \PP_{\mu'} U  =  \findex \PP U  +  \findex \mathbb{Q} U$.  To complete the proof it suffices to show $\findex \mathbb{Q} U = 0$, which is the content of the following proposition.
	\end{proof}
	
	\begin{prop}\label{prop:index is zero in mobility gap}
		If $Q$ is a spectral projection of $H$ onto any subset of the mobility gap $\Delta$, then $$\findex \mathbb{Q}U = 0\,.$$
	\end{prop}
	
	The formal proof of \cref{prop:index is zero in mobility gap} is somewhat technical, and we give it in \cref{sec:proof of index zero in mob gap} below. The main idea is that the range of $Q$ is spanned by a basis $\{\psi_n\}_n$ of eigenvectors of $H$ with good localization properties, namely a so-called SULE basis with polynomial decay (this follows from the fact that $\Delta$ is a mobility gap; see \cref{sec:The SULE basis} for a definition).  We define a unitary operator on $\operatorname{ran}Q$ via
	$$ V \psi_n := \exp(\ii\arg(x_n\cdot e_1+\ii x_n\cdot e_2)) \psi_n\, , $$
	where $x_n$ is a suitably chosen \emph{localization center} for $\psi_n$.  Extending $V$ to the whole Hilbert space $\HH$ as $V\psi=\psi$ for $\psi\in \operatorname{ran}Q^\perp$, we find that $V$ commutes with $Q$ and thus that $\findex\mathbb{Q} V = 0$.  For vectors $\psi\in \operatorname{ran}Q$, $V\psi$ is an approximation to $U\psi$ wherein the flux insertion operator is modified to act at the localization centers of the SULE basis.   As we show below, it follows that $(U-V)Q$ is compact.  Thus $\findex \mathbb{Q}U=\findex \mathbb{Q} V=0$.

	\subsection{The edge topological index and the bulk-edge correspondence}
	We now turn to the formulation of a suitable regularized index for an edge system with a bulk mobility gap $\Delta$, and the equality of this index with the bulk index. To motivate our proposed index, it is useful first to  consider a system with a bulk \emph{spectral gap} $\Delta$.  In that case, the following edge index was formulated in \cite[eq.\ (2.4)]{Fonseca2020} in both the IQHE and the Fu-Kane-Mele $\ZZ_2$ case. Let $\hat{H}$ be an edge Hamiltonian that has a bulk-spectral-gap on $\Delta$ as in \cref{def:edge system with bulk gap}. Then its associated index is
	\begin{align}\label{eq:edge index in spectral gap regime}\hat{\calN}_{(2)}(\hat{H},\Delta) := \findex_{(2),\HHE} \mathbb{W}_1 g(\HE)   \, ,\end{align}
	where $g:\RR\to[0,1]$ is any smooth function with $g(t)\equiv 1$ for $t<\Delta$ and $g(t)\equiv 0$ for $t>\Delta$, and so, in particular, \begin{align}\label{eq:property of g}\supp(g^2-g)\subseteq\Delta\end{align}and where we use the abbreviated "winding" notation \begin{align}\WW_1 A := \mathbb{\Lambda}_1 \exp(-2\pi\ii A)\equiv \Lambda_1\exp\left(-2\pi\ii A\right)\Lambda_1+\Id-\Lambda_1\,.\end{align} The subscript $\HHE$ on $\findex_{(2)}$ is placed to emphasize the the dimensions of kernels are now calculated on the edge Hilbert space $\HHE$. 
	
	$\hat{\calN}$ was first defined by \cite{SBKR_2000}, where it is connected with the edge Hall conductivity in the spectral gap regime. Heuristically, this index counts the number of states pushed across the fiducial line $x_1=0$ by the unitary evolution $\exp(-2\pi \ii g(\hat{H})) $ (or the parity of this number, in the $\ZZ_2$ case). If $H$ is a bulk Hamiltonian with a spectral gap on $\Delta$, $g(H)$ is a spectral projection and $\exp(-2\pi\ii g(H))=\Id$. Moreover, if $\HE$ is an edge Hamiltonian with a bulk-spectral-gap on $\Delta$ as in \cref{def:edge system with bulk gap}, ${\exp(-2\pi \ii g(\hat{H}))-\Id \in \LOC{2}}$, from which it follows \cite[Prop A.4]{Fonseca2020} that $\WW_1 g(\HE)\in\calF$ and hence $\hat{\calN}$ is well defined. 
	
	Back to the case at hand, let now $\HE$ be an edge Hamiltonian with a bulk-mobility-gap on $\Delta$ as in
	\cref{def:edge system with bulk mobility gap}. That implies there exists some bulk Hamiltonian $K$ with a bulk mobility gap on $\Delta$ as in \cref{def:deterministic insulator} and that $\HE-\Ad_\iota K \in \LOC{2}\,.$ In this scenario, we no longer have $\exp(-2\pi\ii g(K))=\Id$, since $K$ possesses eigenvalues $\lambda\in\Delta$ for which $g(\lambda)\neq0,1$, and furthermore, $\exp(-2\pi \ii g(\hat{H}))-\Id$ will not decay into the bulk since it will evolve localized bulk states infinitely far into away from the boundary. 
		
	To formulate a mobility-gap regularized edge index, we proceed with the idea that for the winding, we only want to count ``edge states'' which do not overlap strongly with such localized eigenstates of the bulk. The idea then is to define the regulator \begin{align}\label{eq:definition of the regulator} \hat{R} := \Ad_\iota \chi_{\Delta^c}(K)\,. \end{align} Now we replace the operator $ g(\hat{H})$ by $\hat{R} g(\hat{H}) \hat{R}$, and define the \emph{mobility-gap regularized edge index} as
	\begin{align}\hat{\calN}_{(2)}(\HE,\Delta) \ := \ \findex_{(2),\HHE} \WW_1 \hat{R} g(\HE) \hat{R} \,. \label{eq:edge index} \end{align} 
	
	Below in \cref{sec:edge index well-defined} we will show this index does not depend on the choice of the bulk Hamiltonian $K$ defining the regulator $\hat{R}$, and thus is intrinsic to the edge Hamiltonian $\hat{H}$. As explained above, for spectrally-gapped systems this choice of regulator reduces to $\hat{R}=\Id$. The idea to define an edge index which apparently depends on a bulk object has precedence in the mobility gap regime for the IQHE in \cite{EGS_2005} (although there there was also an independent definition without a bulk Hamiltonian using time-averaging). The definition is ultimately justified using the fact the index does not depend on the choice of $K$, which will be established in \cref{sec:edge index well-defined}.

	
	A main result in this paper is
	\begin{thm}[Bulk-edge correspondence]\label{thm:bec}
		Let $H$ be a mobility gapped insulator on $\Delta$ as in \cref{def:deterministic insulator} and $\HE$ be an edge Hamiltonian with bulk-mobility-gap on $\Delta$ as in \cref{def:edge system with bulk mobility gap}, Furthermore, assume that $H$ and $\HE$ are compatible in the sense that $$ \HE-\Ad_\iota H \in \LOC{2}\,. $$ 
		Then bulk index of $H$ and edge index of $\HE$ agree, namely, $$ \calN_{(2)}(H,\Delta) = \hat{\calN}_{(2)}(\HE,\Delta)\,. $$
	\end{thm}
	This result is new in the context of mobility-gapped TRI $\ZZ_2$ topological insulators. In the IQHE it was proven in \cite[Theorem 1]{EGS_2005} and our proof provides an alternative route, although our edge index in the IQHE is defined differently; see our comments below in \cref{subsec:equivalence of IQHE edge indices}.
	
	\subsection{The edge index is well-defined}\label{sec:edge index well-defined}
	While there are proofs that the bulk index remains well-defined in the mobility gap \cite{Bellissard_1994,Aizenman_Graf_1998}, the edge index has not been similarly studied in the generality considered here.  Unfortunately the results of \cite[Prop A3, A4]{Fonseca2020} do not apply in this case, but we will prove that
	\begin{thm}\label{thm:Fredholm homotopies in the mobility gap regime}
		If $A\in\qProj{2}$ then $\WW_1 A$ is Fredholm.
		Furthermore, if $[0,1]\ni t\mapsto A(t)\in\qProj{2}$ is a norm-continuous family of such operators then $[0,1]\ni t\mapsto\findex\WW_1 A$ is locally-constant.
	\end{thm}
	The proof of \cref{thm:Fredholm homotopies in the mobility gap regime} can be found in \cref{sec:proof of mobility gapped homotopies}.  Here we use it to derive our main result, \cref{thm:bec} as well as the fact that the regularized edge index is well-defined:
	\begin{prop}
		With $\HE$ an edge Hamiltonian with a bulk mobility gap on $\Delta$ as in the preceding section, and with the choice of regulator $\hat{R}$ as in \cref{eq:definition of the regulator}, we have \begin{align}\WW_1 \hat{R} g(\HE) \hat{R}\in\calF(\HHE)\end{align} and \begin{align}\findex_{(2)} \WW_1 \hat{R} g(\HE) \hat{R} = \findex_{(2)} \WW_1 \hat{R'} g(\HE) \hat{R'}\end{align} where $$\hat{R'}\equiv\Ad_\iota\chi_{\tilde{\Delta}^c}(K')$$ is a regulator associated with any other choice $K'$ of bulk mobility gapped insulator on $\Delta$ obeying \cref{def:deterministic insulator} such that $$\HE-\Ad_\iota K'\in\LOC{2}$$ and $\tilde{\Delta}\subseteq\Delta$ is any sub-interval of the mobility gap for which $\supp(g^2-g)\subseteq\tilde{\Delta}$.
		
		Consequently, $\hat{\calN}_{(2)}$ is well-defined and independent of the regularization.
	\end{prop}
	\begin{proof}
		Within this proof we will drop the hats from edge operators for convenience (all operators will be edge operators, including $H\equiv \HE$ within this proof). We will also denote $G := g(H)$ for brevity.
		
		To show $\WW_1 R G R\in\calF$, by \cref{thm:Fredholm homotopies in the mobility gap regime}, it suffices to show that \begin{align}\label{eq:RGR is a 2-quasi-projection} RGR \in\qProj{2}\,. \end{align}  By construction, it is clear that, $R$ is self-adjoint and weakly-local. Furthermore, $R\in\qProj{2}$. Indeed, denoting by $\cong$ equivalence up to terms in $\WLOCC{2}$, we have, using $\iota\iota^\ast=\Lambda_2$ and $\iota^\ast\Lambda=\iota^\ast$, $$ R^2 = \iota^\ast \chi_{\Delta^c}(K)\iota\iota^\ast\chi_{\Delta^c}(K)\iota = \iota^\ast \chi_{\Delta^c}(K)\Lambda_2\chi_{\Delta^c}(K)\iota = \iota^\ast \left([\chi_{\Delta^c}(K),\Lambda_2]\chi_{\Delta^c}(K)+\Lambda_2\chi_{\Delta^c}(K)\right)\iota\cong R $$ where the last equality is due to the fact that $$\partial_2\WLOC\subseteq\WLOCC{2}\,.$$ 
		
		Since $R\in\qProj{2}$, we have $$ (RGR)^2-RGR \cong RGRGR-RGR \cong R(GRG-GR)R = -RGRG^\perp R$$ where $G^\perp := \Id-G$ even though $G$ is \emph{not} a projection. Next, we claim $G-\Ad_\iota g(K)\in\WLOCC{2}$. Indeed, we assume $H-\Ad_\iota K \in \LOC{2}$ and then we use the Helffer-Sj\"ostrand formula to relate $g$ of the operators in terms of their resolvents. Hence, $$ RGRG^\perp R \cong R(\Ad_\iota g(K))R(\Id-\Ad_\iota g(K)) R \cong \iota^\ast \chi_{\Delta^c}(K)(g(K)^2-g(K))\iota = 0\,,$$ the last equality is true since it is assumed (in \cref{eq:property of g}) precisely that $\supp(g^2-g)\subseteq\Delta$. This concludes the proof of \cref{eq:RGR is a 2-quasi-projection}.
		
		Before turning to the proof that $\hat{\calN}$ is independent of the choice of $K$, let us show that the index stays the same if $R\equiv\Ad_\iota\chi_{
		\Delta^c}(K)$ is replaced by $\Ad_\iota\chi_{
		[a,b]^c}(K)$ where we assume that if $\supp(g^2-g)\subseteq[c,d]$ then $$ \inf \Delta < a < c < d < b < \sup\Delta\,. $$ We will use the following fact, proven later in \cref{lem:homotopy tool}: $$ A-B\in\WLOCC{2}\Longrightarrow \findex\WW_1 A = \findex\WW_1 B\,. $$ We change only the starting point of the interval, so that $$ R := \Ad_\iota\chi_{
		[a,b]^c}(K)\,;\,R' := \Ad_\iota\chi_{
		[a+\ve,b]^c}(K)  $$ Hence $$ RGR-R'GR' = RG(R-R')+(R-R')GR' $$ and $R-R' = \Ad_\iota\chi_{[a,a+\ve)}(K)$ whereas $\left.g\right|_{[a,a+\ve)}=1$ by assumption. Hence, using again $G\cong\Ad_{\iota}g(K)$, 
		$$ RGR-R'GR'\cong \Ad_\iota \chi_{[a,a+\ve)}(K)\,. $$ Later on in \cref{subsec:BEC} we will see that, using \cref{lem:flux insertion connected to Kitaev index}, for any bulk projection $Q$, one has $$ \findex_{\HHE} \WW_1 \Ad_\iota Q = \findex_{\mathcal{H}} \mathbb{Q}U $$ and we apply this with $Q:=\Ad_\iota \chi_{[a,a+\ve)}(K)$ to learn that, since $Q$ is fully localized (as $[a,a+\ve)\subseteq\Delta$), we have using \cref{prop:index is zero in mobility gap}, that $\findex \mathbb{Q}U = 0\,.$ But since $\left(\WW_1 RGR\right)\left( \WW_1 R'GR'\right)^\ast\cong \WW_1(RGR-R'GR')$ (here $\cong$ means up to compact operators), this concludes the proof that the end points of $\Delta$ in the choice of $R$ do not affect the index. 
		
		Finally, we come to the invariance under the choice of $K$. Let $K'$ be an alternative choice, i.e., such that also $\HE-\Ad_\iota K' \in \LOC{2}$ and $K'$ has a mobility gap on $\Delta$ as in \cref{def:deterministic insulator}. Using \cref{lem:SULE implies Green function decay} further below, we know that there exists two full measures sets $S,S'\subset \Delta$, dependent on $K,K'$ respectively, such that the resolvents of the respective operators with spectral parameter with real part in $S,S'$ respectively are $\WLOC$, \cref{eq:decay of G from SULE}. Pick $a<b\in\Delta$ such that, without loss of generality, $a<c\in S\cap S'$ and $d<b\in S\cap S'$ where $\supp(g^2-g)\subseteq[c,d]\subseteq\Delta$, and define $$R := \Ad_\iota\chi_{[a,b]^c}(K)\,;\, R' := \Ad_\iota\chi_{[a,b]^c}(K')\,.$$ Indeed, by the proof just before, we know that $\hat{\calN}$ does not change under the replacement $\Delta \to [a,b]$ in $R,R'$.

		Using again \cref{lem:homotopy tool}, it suffices to show that $$ R-R' \in \WLOCC{2}\,. $$ 
		
		We note that we also have $K-K'\in\LOC{2}$ and hence \begin{align*}
			R-R' &= \Id-\Ad_\iota\chi_\Delta(K)-\Id+\Ad_\iota\chi_\Delta(K')\\
			&= \Ad_\iota \frac{\ii}{2\pi}\int (K-z\Id)^{-1}(K-K')(K'-z\Id)^{-1}\dif{z}\,.
		\end{align*}
	
		Now if we pick the contour integral to be the rectangular path in $\CC$ with two horizontal legs parallel to $[a,b]$ (above and below it at distance $1$ from it, say) and which passes vertically through $a$ and $b$. Then using \cref{lem:SULE implies Green function decay} we may conclude that since the factor $K-K'$ is $\WLOCC{2}$ and the two resolvents remain $\WLOC$ into the real axis, by the ideal proprety of $\WLOCC{2}$ within $\WLOC$ we are finished.
	\end{proof}
	
	\subsection{The bulk-edge correspondence proof}\label{subsec:BEC}
	In this section we prove our main theorem, \cref{thm:bec}. A crucial first step is the passage from the bulk index defined in \cref{eq:flux insertion index} (we refer to it as the flux insertion index) to Kitaev's index \cite[(131)]{kitaev2006}, the latter being much more suggestive of an edge geometry.

	\begin{lem}\label{lem:flux insertion connected to Kitaev index} Let $P=\chi_{(-\infty,\mu)}(H)$ and $U=\exp(\ii\arg (X_1 + \ii X_2))$, as above.  Then
		\begin{equation}\findex \PP U \ = \ \findex \WW_1 P \Lambda_2 P \label{eq:flux insertion equals kitaev}\end{equation}
		where we use the notation $\WW_1 A := \mathbb{\Lambda}_1 \exp(-2\pi\ii A) $ for any operator $A$. 
	\end{lem}
	\begin{rem*}
		The right hand side is \emph{Kitaev's index}.  Note that it is well-defined even in the mobility gap regime, due to the fact that $P\Lambda_2 P \in \qProj{2}$.
	\end{rem*}
	\begin{proof}
		
		In the spectral gap regime, the proof of this result may be found in \cite[Thm. 3.1]{Fonseca2020}. That proof remains valid in the mobility gap regime for the IQHE ($\ZZ$ index). However, in the $\ZZ_2$ mobility gap case, we do not know that there are only two path-connected components of TRI topological insulators. This is known in the spectral gap regime, and is used in \cite[Thm. 3.1]{Fonseca2020} to assert that two index formulas are equal globally provided they agree on one trivial and one non-trivial system. Here instead we proceed by a direct homotopy. 
		
		Let $U_a=\exp(\ii\arg (X_1-a_1 + \ii (X_2-a_2))$ be the flux insertion at position $a \in (\ZZ^2)^*$. This can be norm-continuously deformed into $U_{a, R}$, the flux insertion at $a$ with the corresponding phase $u_{a, R}$ different from $1$ only in the cone with vertex at $a$ and opening to the right with opening angle $\nu$. Inspired by the picture where the flux is inserted dynamically, we say that $U_{a, R}$ is flux insertion with electric field supported in the cone to the right of $a$.  Similarly, we will denote by $U_{a, L}$ the flux insertion at $a$ with electric field supported in the cone to the left of $a$. 
		
		The deformation from $U_a$ to $U_{a, R}$ can be done through a norm continuous path of gauge transformations, all of which satisfy the assumptions of \cref{lem:flux_piercing}. It follows that we obtain a norm continuous path of Fredholm operators interpolating from $\PP U_a$ to $\PP U_{a, R}$. We conclude that the bulk index is given by $\calN = \findex \PP U_{a, R}$.

		\emph{Step 1: Introducing a vertical cut:} Choose $a$ far to the right of the fiducial line $x_1 = 0$. We will show that the operator $\bbLambda_1 \PP U_{a, R}$ is Fredholm, with the same index as $\PP U_{a, R}$. To do this, we simply show that the difference is compact. First note that
		$$\bbLambda_1\PP U_{a,R} \ = \ \Lambda_1 \PP U_{a,R} \Lambda_1 + \Lambda_1^\perp \ = \ \Lambda_1 \PP U_{a,R} \Lambda_1 +
		\Lambda_1^\perp \PP U_{a,R} \Lambda_1^\perp - \Lambda_1^\perp (\PP U_{a,R} -\Id)\Lambda_1^\perp \ .$$
		Here 
		$$ \Lambda_1^\perp (\PP U_{a,R} -\Id)\Lambda_1^\perp \ = \ \Lambda_1^\perp (P [ U_{a,R},P] + P(U_{a,R} - \Id))\Lambda_1^\perp \ = \ \Lambda_1^\perp  P [ U_{a,R},P] \Lambda_1^\perp \ $$
		is compact since  $[U_{a, R}, P]$ is compact by \cref{lem:flux_piercing} and $U_{a, R} - \Id$ is a multiplication operator supported on the cone to the right of $a$, which is disjoint from the support of $\Lambda_1^\perp$. In a similar way, we find that $\Lambda_1\PP U_{a,R} \Lambda_1^\perp$ and $\Lambda_1^\perp \PP U_{a,R} \Lambda_1$ are compact. 
		We conclude that $\bbLambda_1\PP U_{a,R}  - \PP U_{a,R} $ is compact and $\calN = \findex \bbLambda_1 \PP U_{a, R}$.

		\emph{Step 2: Inserting an opposing flux:} We insert an opposite flux $U_{-a, L}^*$ at the point $-a$ far to the left of the fiducial line $x_1 = 0$, with electric field supported in the cone to the left of $-a$. Having fixed $a$, we  henceforth drop it from the notation, so $U_R = U_{a, R}$ and $U_{L} = U_{-a, L}$. We denote the combined flux insertion by  $U_{LR} := U_L^* U_R$. We show that the corresponding operator $\bbLambda_1 \PP U_{LR}$ is still Fredholm with the same index as $\bbLambda_1 \PP U_{a, R}$. Indeed, their difference is compact:
		$$\bbLambda_1 \PP U_{LR} - \bbLambda_1 \PP U_R \ 
		\ = \ \Lambda_1 P (U_L^*-\Id ) U_R P \Lambda_1 \ = \ \Lambda_1 (U_L^* - \Id) P \Lambda_1 + \Lambda_1 [P, U_L^*] P \Lambda_1 \ , $$
		where $\Lambda_1 (U_L^* - \Id) = 0$ because the multiplication operators have disjoint supports and $[P, U_L^*]$ is compact by \cref{lem:flux_piercing}.

		\begin{figure}
			\begin{center}
				\begin{tikzpicture}[scale=0.5]
					\def\a{4};
					\draw[fill=black] (-\a,0) circle (3pt);
					\draw[fill=black] (\a,0) circle (3pt);
					\draw[thick,dotted] (-\a,0) -- (\a,0);
					\node  at (-\a,-.5) {$-a$};
					\node  at (\a,-.5) {$a$};
					\node at (0,{-0.5*\a}) {$\xi(x)=0$};
					\node at (0,{0.5*\a}) {$\xi(x)=2\pi$};
					\draw [blue,thick,dotted,domain=-30:6] plot ({\a+0.75*\a*cos(\x)}, {0.75*\a*sin(\x)});
					\draw [orange,thick,dotted,domain=-6:30] plot ({-\a-0.75*\a*cos(\x)}, {-0.75*\a*sin(\x)});
					\draw [red,thick,dotted,domain=-30:30] plot ({\a+2*\a*cos(\x)}, {2*\a*sin(\x)});
					\draw [thick,domain=0:{2*\a+0.5}] plot ({\a + \x*cos(30)},{\x*sin(30)});
					\draw [thick,domain=0:{2*\a+0.5}] plot ({\a + \x*cos(30)},{-\x*sin(30)});
					\draw [thick,domain=0:{2*\a+0.5}] plot ({-\a - \x*cos(30)},{\x*sin(30)});
					\draw [thick,domain=0:{2*\a+0.5}] plot ({-\a - \x*cos(30)},{-\x*sin(30)});
					\draw [red,thick,dotted,domain=-30:30] plot ({-\a-2*\a*cos(\x)}, {2*\a*sin(\x)});
					\draw [blue,domain=0:{\a}] plot ({\a +\x*cos(6)},{\x*sin(6)});
					\draw [orange,domain=0:{\a}] plot ({-\a -\x*cos(6)},{\x*sin(6)});
					\draw [fill=blue] ({\a +\a*cos(6)},{\a*sin(6)}) circle (3pt);
					\draw [fill=orange] ({-\a -\a*cos(6)},{\a*sin(6)}) circle (3pt);
					\node at ({\a +\a*cos(6) +1},{\a*sin(6)+0.75}) {$\xi(x)=\frac{2\pi}{\nu}\theta$};
					\node at ({-\a -\a*cos(6) -1},{\a*sin(6)+0.75}) {$\xi(x)=\frac{2\pi}{\nu}\phi$};
					\node at ({\a+0.75*\a*cos(-12)+0.5}, {0.75*\a*sin(-12)}) {\textcolor{blue}{$\theta$}};
					\node at ({-\a-0.75*\a*cos(12)-0.5}, {-0.75*\a*sin(12)}) {\textcolor{orange}{$\phi$}};
					\node at ({3*\a+0.5}, 0) {\textcolor{red}{$\nu$}};
					\node at ({-3*\a-0.5}, 0) {\textcolor{red}{$\nu$}};
				\end{tikzpicture}
			\end{center}
			\caption{The double flux insertion function $\xi(x)$. }
			\label{figure 1}
		\end{figure}
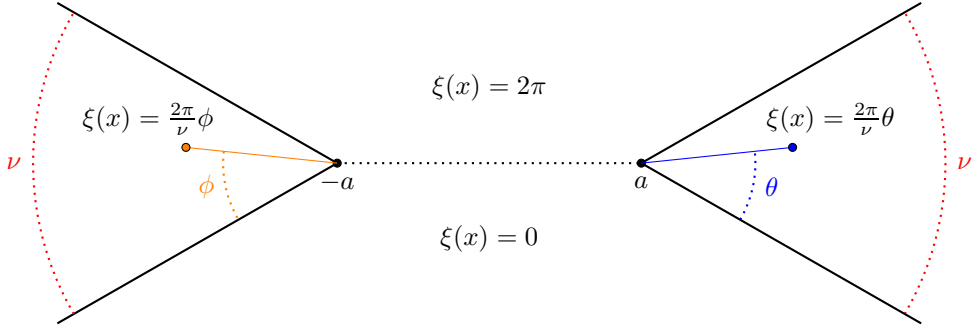

		\emph{Step 3: Raising the Fermi projection to the exponent:} The operator $U_{LR}$ can be written as $U_{LR} = \ee^{\ii \xi}$ where $\xi$ is a multiplication operator that takes the value $2 \pi$ on the upper half plane outside of the cones, the value 0 on the lower half plane outside of the cones, and interpolates continuously on the cones (see \cref{figure 1}). We thus have $\calN = \findex \bbLambda_1 \PP \ee^{\ii \xi} $.

		We now want to `raise $P$ to the exponent,' i.e. we want to relate $\calN$ to the index of
		\begin{align*}
			\bbLambda_1 \PP \ee^{\ii P \xi P} = \bbLambda_1 (P \ee^{\ii P \xi P} P + P^\perp) = \bbLambda_1 \ee^{\ii P \xi P}.
		\end{align*}
		The difference with our starting point is
		\begin{align*}
			\bbLambda_1 \PP \ee^{\ii \xi} - \bbLambda_1 \ee^{\ii P \xi P} = \Lambda_1 ( P \ee^{\ii \xi} P - P \ee^{\ii P \xi P} P ) \Lambda_1.
		\end{align*}
		We will show that this difference is compact. Indeed
		\begin{equation}\label{eq:expanded exponential}
			P \ee^{\ii \xi} P - P \ee^{\ii P \xi P} P \ = \ \sum_{n \geq 2} \frac{\ii^n}{n!} \left( P\xi^n P - (P \xi P)^n \right) 
			\ = \ \sum_{n \geq 2} \frac{\ii^n}{n!} \sum_{\substack{\{\sharp_i\} : \sharp_i \in \{1, \perp\} \\ \text{at least one } \perp} }   P \xi P^{\sharp_1} \xi \cdots \xi P^{\sharp_{n-1}} \xi P \ .
		\end{equation}

		Each of these terms contains at least one factor $P \xi P^\perp$, which we claim is compact. To see this, note that
		$[P, \xi] = P \xi P^\perp - P^\perp \xi P $, so it is sufficient to show that $[P, \xi]$ is compact. This follows from \cref{lem:flux_piercing} since  we may choose $\xi$ to satisfy $\norm{\xi}_\infty=2\pi$ and
		$$|\xi(x+b)-\xi(x)| \ \le \ \frac{1}{\nu} \frac{\norm{b}}{\norm{x} -\norm{a}} \ $$
		for $\norm{x}>\norm{a}$.  Thus the right hand side of \cref{eq:expanded exponential} is a norm-convergent sum of compact operators, and thus compact.  We conclude that $\calN=\findex \bbLambda_1 \ee^{\ii P \xi P}$.

		\emph{Step 4: Flattening the curve:} Finally we want to transform the multiplication operator $\xi$ to (a multiple of) the half-space projector $-2 \pi \Lambda_2$. This can be done by continuously closing the cones, taking $\nu\rightarrow 0$, which yields a norm continuous deformation from $\bbLambda_1 \ee^{\ii P \xi P}$ to $\bbLambda_1 \ee^{2 \pi \ii P \Lambda_2 P }$. As long as the cones have not closed completely, the above arguments show that the interpolating operators are Fredholm. The endpoint of the interpolation is shown to be Fredholm in \cref{thm:Fredholm homotopies in the mobility gap regime}. We conclude that
		$$
		\calN \ = \ \findex \bbLambda_1 \ee^{-\ii P \Lambda_2 P} \ . \qedhere
		$$
	\end{proof}

	To complete the proof of \cref{thm:bec}, we will need the following basic lemma:
	\begin{lem}\label{lem:homotopy tool}
		If $A,B\in\qProj{2}$ such that $A-B\in\WLOCC{2}$ then $\findex \mathbb{W}_1 A = \findex \mathbb{W}_1 B$.	\end{lem}
	\begin{proof}
		Consider the homotopy $ [0,1]\ni t \mapsto  t(A-B)+B$. We have $$ (t(A-B)+B)^2-t(A-B)-B =  t^2(A-B)^2+t((A-B)B+B(A-B))+B^2-B-t(A-B) $$ so that at every point along the homotopy $t(A-B)+B\in\qProj{2}$.  The result follows from \cref{thm:Fredholm homotopies in the mobility gap regime}.
	\end{proof}
	Using this lemma, we have 
	$$ \findex \WW_1 P \Lambda_2 P = \findex \WW_1 \Lambda_2 P  \Lambda_2\,. $$
	Indeed, it is sufficient to note that $P\Lambda_2 P, \Lambda_2 P \Lambda_2 \in \qProj{2}$ and that $$ P\Lambda_2 P- \Lambda_2 P \Lambda_2 \simeq \Lambda_2[P,\Lambda_2]P^\perp\in\WLOCC{2}\,. $$
	
	Next, note that by definition $Q g(H) Q = P$ if we take the Fermi energy at the bottom of $\Delta$, and then we can use invariance with respect to the Fermi energy, \cref{thm:stability of index}, to move to it any other point of $\Delta$, so that we may write $$\findex \WW_1 \Lambda_2 P  \Lambda_2 = \findex \WW_1 \Lambda_2  Q g(H) Q \Lambda_2 \,.$$
	Also $$\findex \WW_1 \Lambda_2  Q g(H) Q \Lambda_2 = \findex \WW_1 \Lambda_2  Q g(\Lambda_2 H\Lambda_2) Q \Lambda_2. $$  This follows from the above lemma, thanks to the fact that $\Lambda_2  Q g(H) Q \Lambda_2,   \Lambda_2  Q g(\Lambda_2 H\Lambda_2) Q \Lambda_2 \in \qProj{2}$ and that $$ \Lambda_2  Q g(H) Q \Lambda_2 -   \Lambda_2  Q g(\Lambda_2 H\Lambda_2) Q \Lambda_2 \in\WLOCC{2}\,. $$ For a proof of this latter fact, see \cite{Elbau_Graf_2002} or \cite[Rem. A9]{Fonseca2020}.
	
	Finally, we have, for Dirichlet boundary conditions (with $\HE = \ad_\iota H $), $$ \ker_\HHB \WW_1 \pm \Lambda_2  Q g(\Lambda_2 H\Lambda_2) Q \Lambda_2 \cong \ker_{\HHE} \WW_1\left( \pm   \hQ g(\HE) \hQ\right) $$ and the transition to any boundary conditions is done by a further interpolation (as in \cite[Prop. A10]{Fonseca2020}). This concludes the proof of \cref{thm:bec}, i.e., that $ \calN = \hat{\calN} $. Here we have shown equivalence with the particular choice of $\hat{R} = \hat{Q}$; in \cref{lem:edge hamiltonian compatible with bulk hamiltonian obeys edge-mob-gap def} we saw that this choice is appropriate according to the definition of the edge index.

	\section{The proof of \cref{thm:Fredholm homotopies in the mobility gap regime} on mobility gap homotopies}\label{sec:proof of mobility gapped homotopies}
	Here we shall prove our main tool \cref{thm:Fredholm homotopies in the mobility gap regime}.
	For any $N\in\NN$, let $p_N$ be the polynomial approximation of $\alpha\mapsto\exp(-2\pi\ii \alpha)=:p_\infty(\alpha)$ up to order $N$: $$ p_N(\alpha) \ = \ \sum_{n=0}^N \frac{1}{n!}(-2\pi\ii\alpha)^n\,. $$ and define $$f_N(\alpha) \ := \  p_N(\alpha)-(p_N(1)-1)\alpha \ =: \ 1+\sum_{n=1}^N \varphi_n \alpha^n\,.$$ We observe that $f_N$ is also a polynomial, and that $f_N(0)=f_N(1)=1$. Moreover, $\lim_N f_N = p_\infty $, uniformly as functions on $\RR$.

	\begin{lem}\label{lem:polynomial fredholm}
		If $A\in\qProj{2}$ is self-adjoint then there is some $N_A\in\NN$ such that for all $N\geq N_A$, $$ \mathbb{\Lambda}_1 f_N(A) \ =: \ B_N $$ is a Fredholm operator.
	\end{lem}
	\begin{proof}
		Since $U := p_\infty(A)$ is unitary, it is invertible. Since the set of invertible operators is open, there is  $\ve>0$ such that $$ B_\ve(U) \subseteq \mathrm{invertibles}\,.$$ In fact, $\ve$ may be estimated using the Neumann series: since $U$ is unitary, one may taken any $\ve<1$. Hence there is some $N$ large enough (indepedent of $A$ due to uniform convergence) so that $\norm{U-f_N(A)}<1$ and hence $f_N(A)$ is invertible. Let $C_N := \mathbb{\Lambda}_1 ((f_N(A))^{-1})$. A short calculation yields $$ \Id-B_N C_N \ = \  -\Lambda_1[f_N(A)-\Id,\Lambda_1]((f_N(A))^{-1})\Lambda_1\,.$$ Now, $f_N(A)-\Id = \sum_{n=1}^{N}\varphi_n A^n$ may be expressed, using the fact that $\sum_{n=1}^N \varphi_n = 0$, as $$f_N(A)-\Id \ = \ \sum_{n=2}^N \varphi_n\sum_{k=0}^{n-2}A^k(A^2-A) \ ,$$ which is evidently the product of a polynomial (of finite degree) of $A$ times $A^2-A\in\WLOCC{2}$.  Thus, the whole expression is in $\WLOCC{2}$. After taking the commutator with $\Lambda_1$ we find an expression which is trace-class, whence $C_N$ is found to be the parametrix of $B_N$ and hence our claim.
	\end{proof}

	Finally we come to our main \cref{thm:Fredholm homotopies in the mobility gap regime}.
	In the proof we will use the following Lemma of Dieudonne:
	\begin{lem}[Dieudonne]
		If $F$ is a Fredholm operator and $G$ is \emph{any} parametrix of $F$, then the open ball $B_{\norm{G}^{-1}}(F)$ lies within the set of Fredholm operators.
	\end{lem}
	\begin{proof}[Proof of \cref{thm:Fredholm homotopies in the mobility gap regime}]
		By \cref{lem:polynomial fredholm} and its proof, there is $N_A$ such that if $N\geq N_A$ then $B_N$ is Fredholm with parametrix $C_N$. Hence using the Lemma of Dieudonne, we see it suffices to show that $$ \norm{B_N-\WW_1 A }\ \leq \ \frac{1}{\norm{C_N}} \ . $$
		In fact we may estimate $\norm{C_N} \leq 2 + \norm{(f_N(A))^{-1}}$ and $\norm{(f_N(A))^{-1}} \leq 2$ if $\norm{f_N(A)-U}<\frac{1}{2}$. Hence we can apparently pick $N$ large enough (independent of $A$) such that $\norm{f_N(A)-U}<\frac{1}{4}$, which suffices.
	\end{proof}
	
	\section{Open questions}
	\subsection{Equivalence of edge indices in the mobility gap regime}\label{subsec:equivalence of IQHE edge indices}
	We have proved the bulk-edge correspondence in the mobility gap regime for both the $\ZZ$-index (IQHE) and $\ZZ_2$ index, using our regularized edge index \cref{eq:edge index}, which is new. For the IQHE in the mobility gap regime, there was already an alternative regularized edge formula for the edge conductance, defined \cite[Eq. (1.12)]{EGS_2005}. A priori, it is not clear that the two quantities are equal.  Of course, a posteriori, they are seen to be equal thanks to our bulk-edge correspondence \cref{thm:bec} and that of \cite{EGS_2005}. However, it is an interesting question whether one can prove this directly, without reference to the bulk index, along the lines of \cite{SBKR_2000}.
	
	\subsection{The nature of the edge spectrum in the mobility gap regime}
	For the IQHE, we conjecture that in the mobility gap regime there is absolutely continuous spectrum in $\Delta$ for $\HE$. Using the results of \cite[Theorem 2.1 3.]{Asch_Bourget_Joye_2002}, we know that having a non-trivial index implies that $$\sigma_{\mathrm{ac}}(\exp(-2\pi\ii \hQ g(\HE)) \hQ)) = \mathbb{S}^1 \ , $$ but it is not clear what this implies about the spectral nature of $\HE$ within $\Delta$ due to the presence of $\hQ$ on the left hand side. 
	
	In \cite{Fröhlich2000,Bievre_Pule_doi:10.1142/9789812777874_0003} it was proved that continuum non-trivial IQHE edge systems have ac-spectrum in the spectral gap regime using the Mourre esetimate, and in \cite{bols2021absolutely} the same was proven for discrete systems via \cite[Theorem 2.1 3.]{Asch_Bourget_Joye_2002} (for them $\hQ=\Id$ as they work in the spectral gap regime). Recently \cite{Bols_Cedzich_2022} proved the same thing for the spectral gap regime of $\ZZ_2$ time-reversal-invariant discrete systems.
	
	\subsection{Local $\ZZ_2$ trace formulas}
	We conjecture that in the spectral gap regime, \begin{align}\hat{\calN} = 2\pi \left(\lim_{T\to\infty}\tr(g'(\HE) (\ii[\Lambda_1^T, \HE])_{+})\right) \mod 2\end{align} where $A_+ \equiv \frac{1}{2}(A+|A|)
	$ denotes the positive part of an operator and $$ \Lambda_1^T \equiv \frac{1}{2T}\int_{-T}^T \exp(\ii t \HE)\Lambda_1\exp(-\ii t\HE)\dif{t}\,. $$ For this formula to make sense, we would need to show that \begin{align}2\pi \left(\lim_{T\to\infty}\tr(g'(\HE) (\ii[\Lambda_1^T, \HE])_{+})\right) \in \ZZ\,.\end{align} Once this has been established it is suggestive to regularize this formula in the mobility-gap regime similarly to \cite{EGS_2005}. However, this formula also includes a limiting process (cf. \cite[Thm. 2.3]{Fonseca2020}), so that it's unclear what value it might hold.

	\appendix

	\section*{Appendices}
	\section{A Schatten class lemma} 
	A fundamental tool for our arguments is the observation that if $P\in \WLOC{}$ and $f(X)$ is a multiplication operator in the position basis such that $|f(x)-f(y)|$ decays suitably as $x,y\rightarrow \infty$, then $[P,f(X)]$ is compact.  In fact we will show that this operator is \emph{Schatten-3}, i.e. $\norm{[P,f(X)]}_3^3 = \tr |[P,f(X)]|^3<\infty$, under natural conditions on $f$.  The proof of this fact follows similar arguments to those presented in \cite{Avron1994220,Aizenman_Graf_1998}.
	\begin{lem}\label{lem:flux_piercing}
		Let $P\in\WLOC$ be such that $\norm{P}\leq 1$ and $f\in \ell^\infty(\ZZ^2)$ be such that \begin{equation} |f(x)-f(y)| \ \leq \  D\frac{\norm{x-y}}{1+\norm{x}} \label{eq:admissible flux functions}\end{equation} with $D<\infty$. Then $ [P,f(X)]$
		is Schatten-$3$.
	\end{lem}
	\begin{proof}
		We have $[P,f(X)]_{xy} = P_{xy}(f(x)-f(y))$ and $$ \norm{[P,f(X)]}_3 \ \leq \  \sum_{b\in\ZZ^2}\left(\sum_{x\in\ZZ^2}\norm{P_{x+b,x}}^3|f(x)-f(x+b)|^3\right)^{1/3}\,. $$ 
		Let $B\subseteq\ZZ^2$ be a finite set, to be specified below. Applying the estimate \cref{eq:weakly-local operator} for $x\in B$ (since $P\in \WLOC$) and, for $x\in B^c$, noting that $\norm{P_{x,x+b}}\leq 1$, we conclude that there is $\nu\in\NN$ such that for any $\mu\in\NN$ there is $C_\mu<\infty$ with which  \begin{equation}\label{eq:flux insertion intermediate equation}
			\sum_{x\in\ZZ^2}\norm{P_{x+b,x}}^3|f(x)-f(x+b)|^3 \ \leq \  2\|f\|_\infty^3 \sum_{x\in B} C_\mu^3 \frac{(1+\norm{x})^{3\nu}}{ (1+\norm{b})^{3\mu}} + \sum_{x\in B^c} D^3 \frac{\norm{b}^3}{(1+\norm{x})^3}\,.
		\end{equation}
		
		Now pick $B$ to be the set of $x$ such that $1+\norm{x}\leq (1+\norm{b})^\alpha$, with $\alpha$ still to be determined.Then we find that the second term on the right hand side of \cref{eq:flux insertion intermediate equation} is bounded above by $$ D^3\norm{b}^3(1+\norm{b})^{-\alpha/2}\sum_{x}(1+\norm{x})^{-2.5} \ \leq \ \tilde{D}^3\norm{b}^3(1+\norm{b})^{-\alpha/2}\, ,  $$ while the first term on the right hand side is bounded above by $$ 8 C_\mu^3(1+\norm{b})^{-3\mu}\sum_{x\in B}(1+\norm{x})^{3\nu} \leq 8 C_\mu^3(1+\norm{b})^{-3\mu+3\nu\alpha}|B|\,. $$
		Since $B$ is a ball about the origin of radius $(1+\norm{b})^\alpha-1$, $B$ is bounded above by $\tilde{C} (1+\norm{b})^{2\alpha}$ for some universal $\tilde{C}<\infty$. Hence the first term on the right hand side of \cref{eq:flux insertion intermediate equation} is bounded above by $$ 8 C_\mu^3 \tilde{\tilde{C}}(1+\norm{b})^{-3\mu+3\nu\alpha+2\alpha}\,. $$
		
		To make the sum $\sum_{b\in\ZZ^2}$ finite we choose $\alpha/2-3>6$ and $3\mu-(3\nu+2)\alpha>6$ ($6$ since we are taking the $1/3$ power and we need at least power, say, $2$ to make this summable on $\ZZ^2$). Both of these may be arranged since $\alpha$ was arbitrary and $\mu$ may be taken arbitrarily large.
	\end{proof}
	
	\section{The SULE basis}\label{sec:The SULE basis}
	In this section let $\mathcal{H}=\ell^2(\ZZ^d)\otimes \CC^N$ and $V\subset \mathcal{H}$ a closed subspace.
	\begin{defn}[SULE basis]\label{def:SULE basis}
		A \emph{Semi-Uniformly Localized basis} for $V$ is an orthonormal basis $\Set{\psi_n}_n$ such that there are a sequence of ``localization centers'' $\Set{x_n}_n\subseteq\ZZ^d$ and $\nu\in\NN$ so that for any $\mu>0$ it holds that \begin{align} \norm{\psi_n(x)} \leq C_\mu(1+\norm{x-x_n})^{-\mu}(1+\norm{x_n})^{\nu}\qquad(x\in\ZZ^d)\, \label{eq:SULE cond}\end{align}
		with $C_\mu <\infty$.
	\end{defn}
	\begin{rem*} When a semi-uniformly localized basis $\Set{\psi_n}_n$ consists of eigenfunctions for a self-adjoint operator, it is called a  Semi-Uniformly Localized Eigenfunction (SULE) basis. This notion was originally defined in \cite{delRioJitomirskayaLastSimon1996}.
	\end{rem*}

	It is shown in \cite[Corollary 7.3]{delRioJitomirskayaLastSimon1996} that the localization centers $\Set{x_n}_n$ obey \begin{align}\sum_{n}\frac{1}{(1+\norm{x_n})^{d+\ve}}\ < \ \infty\qquad(\ve>0) \ , \label{eq:SULE summability}\end{align} a fact we shall use below.  The estimate \cref{eq:SULE cond} of course implies that the operator on $\ell^2(\ZZ^d)$ with matrix elements $\norm{\psi_n(x)}\norm{\psi_n(y)}$ is $\WLOC$.
	
	The following proof appeared in \cite[Section 3.6]{EGS_2005}. We include it here for completeness.
	\begin{lem}
		For an interval $\Delta\subseteq\RR$, if $H$ is a $\Delta$-insulator on $\mathcal{H}$ in the sense of \cref{def:deterministic insulator} then there exists a SULE basis for the vector subspace $\im(\chi_\Delta(H))$ consisting of eigenfunctions for $H$.
	\end{lem}
	\begin{proof}
		We let $P:=\chi_\Delta(H)$ and $P_\lambda := \chi_{\Set{\lambda}}(H)$ for each eigenvalue $\lambda\in \Delta$ of $H$.
		Since $\chi_{\Set{\lambda}}$ is a bounded Borel function on $\Delta$, we conclude from \cref{def:deterministic insulator} that there is $\nu $ such that for any $\mu$ we have \begin{align}\label{eq:SULE est 2} \norm{(P_{\lambda})_{x,x_0}} \ \leq \ C_\mu(1+\norm{x-x_0})^{-\mu}(1+\norm{x_0})^{\nu}\,  \end{align}
		for every eigenvalue $\lambda\in \Delta$, with $C_\mu<\infty$.
		
		Since all eigenfunctions are of finite multiplicity (see \cref{def:deterministic insulator}), we have $\tr(P_\lambda)<\infty$, so $a_x=\norm{(P_\lambda)_{xx}} \le \tr(P_\lambda)_{xx}$ is a summable sequence.  Let $x_0\in \ZZ^d$ be a point at which $a_x$ attains its maximum value and  let $\vec{v}_0\in \CC^N$ with $\norm{\vec{v}_0}=1$ and $\vec{v}_0^\dagger (P_\lambda)_{xx}\vec{v}_0=a_{x_0}$. Now define $$ \psi(x) \ :=\ \frac{1}{\sqrt{a_{x_0}}} (P_{\lambda})_{x,x_0}\vec{v}\,. $$
		One verifies that $\psi$ is an eigenvector for $H$ with eigenvalue $\lambda$, and it is normalized so that $\norm{\psi}^2=1$. 
		We have the bound \begin{multline} \norm{(P_{\lambda})_{x,x_0}\vec{v_0}} \ = \ \max_{\norm{\vec{v}}=1}  |\ip{\delta_x\otimes \vec{v}}{P_\lambda\delta_{x_0}\otimes \vec{v}_0}| \ = \ \max_{\norm{\vec{v}}=1} |\ip{P_\lambda \delta_x\otimes \vec{v}}{P_\lambda\delta_{x_0}\otimes \vec{v}_0}| \\ \leq\  \sqrt{a_x}\sqrt{a_{x_0}} \ \leq \ a_{x_0} \label{eq:SULE est 1}\end{multline} where in the penultimate step the Cauchy-Schwarz inequality was used and in the last step the fact that $a$ achieves its maximum at $x_0$. Combining \cref{eq:SULE est 1,eq:SULE est 2} and noting that $a_{x_0}\le 1$, we find that $$ \norm{\psi(x)} \ \leq \ \sqrt{C_\mu} (1+\norm{x-x_0})^{-\mu/2}(1+\norm{x_0})^{\nu/2}\,. $$
		
		Applying the same process again now to $P_\lambda-\psi\otimes\psi^\ast$, whose rank is smaller by $1$ compared to $P_\lambda$, we obtain the result by induction.
	\end{proof}

    One further consequence of our definition of an insulator (\cref{def:deterministic insulator}) is that matrix elements of the resolvent decay, as expressed in the following. 
	\begin{lem}\label{lem:SULE implies Green function decay}
		If $H$ is a $\Delta$-insulator as in \cref{def:deterministic insulator} then there is a (fixed, $H$-dependent) subset $S\subseteq\Delta$ of full Lebesgue measure such that, for any $E\in S$ one has \begin{align}\label{eq:decay of G from SULE} \left(H-\left(E+\ii\ve\right)\Id\right)^{-1} \in\WLOC  \end{align} with $\WLOC$ estimates uniform in $\ve\in[-1,1]\setminus\Set{0}$.
	\end{lem}
	\begin{rem}
	For random operators, the decay manifested in \cref{eq:decay of G from SULE} is an almost-sure consequence of the various methods used to prove localization, and so, in principle could have been included in the definition of a deterministic insulator. We preferred to keep that definition however as is and provide the deterministic proof below.
	\end{rem}
	\begin{proof}
	Let $R(z)=(H-z)^{-1}$ for $\Im{z} >0$.  It suffices to show, for any closed interval $\Delta'\subset \mathrm{Int}\Delta$ (the interior of $\Delta$),  that there is a full measure set $S'\subset \Delta'$ such that \cref{eq:decay of G from SULE} holds for $E\in S'$.  Fix $\Delta'$ and let $\phi:\RR\rightarrow [0,1]$ be a smooth function with $\phi(x)=1$ on $\Delta'$ and $\phi(x)=0$ on $\Delta^c$.  Then 
	$$ R(z) \ = \ \phi(H) (H-z)^{-1} \ + \ (1-\phi(H)) (H-z)^{-1} \ =: \ R_\phi(z) + R_{1-\phi}(z) \ . $$
	We will bound $R_\phi\in \WLOC$ and $R_{1-\phi}\in \WLOC$ separately. The result for $R_{1-\phi}$ is immediate, since $\sigma(H)\ni\lambda\mapsto(1-\phi(\lambda))(\lambda-E-\ii\ve)^{-1}$ is a smooth function with derivatives bounded uniformly in $\ve\in[0,1]$, so for $E\in \mathrm{Int}\Delta'$ we have $R_{1-\phi}(E+\ii\ve) \in \LOC{}\subseteq\WLOC$, uniformly in $\ve\in [0,1]$. To bound $R_\phi$ we will use the fact that $\Delta$ is a mobility gap. We have
	$$ R_\phi(z) \ = \ \sum_{\lambda \in \mathcal{E}(\Delta)} \phi(\lambda) (\lambda-z)^{-1} P_\lambda \ , $$
	where $\mathcal{E}(\Delta)$ denotes the (countable) set of eigenvalues of $H$ in $\Delta$. 
	
	For $j=1,\ldots,N$ let $\vec{e}_j$ denote the elements of the standard basis for $\CC^N$.  
	Fix $x,y\in\ZZ^2$ and $i,j\in \{1,\ldots,N\}$ and let
	$$f_{x,y}^{i,j}(z)\ := \ \ip{\vec{e}_i}{R_\phi(z)_{x,y}\vec{e_j}} \ = \ \sum_{\lambda \in \mathcal{E}(\Delta)} \phi(\lambda) (\lambda-z)^{-1} \ip{\vec{e}_i}{(P_\lambda)_{x,y}\vec{e}_j} \ = \ \int_{\Delta} (\lambda-z)^{-1} \phi(t) \dif{m_{x,y}^{i,j}}(t)$$
	with 
	$$\dif{m_{x,y}^{i,j}}(t) \  := \ \sum_{\lambda \in \mathcal{E}(\Delta)} \ip{\vec{e}_i} {(P_\lambda)_{x,y}\vec{e}_j}
	\, \delta(t-\lambda)\dif{\lambda} \ .$$
	Because $f_{x,y}^{i,j}$ is the Borel transform of the finite Borel measure $m_{x,y}^{i,j}$, it is well known that the limit $\lim_{\ve\to 0} f_{x,y}^{i,j}(E+\ii\ve)$ exists for almost every $E$ and satisfies
	\begin{equation}\label{eq:weakL1} \left |\{E \in \RR \ : \  |f_{x,y}^{i,j}(E+\ii 0)| > \alpha \} \right | \ \le \ \frac{C}{\alpha}\int_{\RR} \phi(t)\dif{|m_{x,y}^{i,j}|(t)} \ , 
	\end{equation}
	where $| \cdot |$ denotes Lebesgue measure, $\dif{|m_{x,y}^{i,j}|}(t) = \sum_{\lambda} |\ip{\vec{e}_i}{(P_\lambda)_{x,y}\vec{e}_j}| \, \delta(t-\lambda) \dif{\lambda}$ is the total variation measure for $\dif{m_{x,y}^{i,j}}$, and $C$ is a universal constant.  To see this, recall that 
	$\lim_{\ve \to 0} \int \Im{\frac{1}{t-E-\ii\ve}}\dif{m_{x,y}^{i,j}}(t)=0$ a.e., since $m_{x,y}^{i,j}$ is purely singular, while $\lim_{\ve \to 0} \int \Re{\frac{1}{t-E-\ii\ve}}\dif{m_{x,y}^{i,j}}(t)=0 = Hm_{x,y}^{i,j}(E)$, the Hilbert transform of $m_{x,y}^{i,j}$, a.e..  Thus \cref{eq:weakL1} follows from Loomis's weak $L^1$ bound on the Hilbert transform of a measure \cite{loomis1946}. 
	
	The integral  on the right hand side of \cref{eq:weakL1} may be bounded as follows
	\begin{multline} \label{eq:bound on total variation}
	    \int_{\RR} \phi\dif{|m_{x,y}^{i,j}|} \ = \ \sum_{\lambda\in \mathcal{E}(\Delta)} \phi(t) |\ip{\vec{e}_i}{(P_\lambda)_{x,y}\vec{e}_j}| \ \le \ \left | \sup_{|g|\le 1} \sum_{\lambda\in \mathcal{E}(\Delta)} g(\lambda) \ip{\vec{e}_i}{(P_\lambda)_{x,y}\vec{e}_j} \right | \\ = \ 
	\sup_{|g|\le 1} \ip{\delta_x\otimes \vec{e}_i}{P_\Delta(H) g(H) \delta_y \otimes \vec{e}_j} \ \le \ 
	\sup_{|g|\le 1} \norm{ (P_\Delta(H) g(H))_{x,y}} \ .
	\end{multline} 
	Using \cref{def:deterministic insulator} to bound the right hand side, we see that we have shown the following
	\begin{quote} \emph{There are $\nu \in \NN$ such that for every $\mu \in \NN$ and $\alpha>0$  we have}
	\begin{equation}\label{eq:weakL1poly}
	\left |\{E \in \RR \ : \  |f_{x,y}^{i,j}(E+\ii 0)| > \alpha \} \right | \ \le \ C_\mu (1+\norm{x})^{\nu} (1+ \norm{x-y})^{-\mu} \frac{1}{\alpha } \ .
	\end{equation}
	\end{quote}
	
	To prove \eqref{eq:decay of G from SULE} we need to extend \eqref{eq:weakL1poly} off the real axis. For the moment let $x,y,i,j$ be fixed and write $f\equiv f_{x,y}^{i,j}$, $m\equiv m_{x,y}^{i,j}$ to simplify notation. For this purpose, let $0<s<1$ and note that the function $|f(z)|^s$ is sub-harmonic in the upper half plane. 
	Let $\tilde{\Delta}= \{t : \dist(t,\Delta)\le 1\}$ and define
	$$g(E) = |f(E+\ii 0)|^s \chi_{\tilde{\Delta}}(E) \quad \text{and} \quad h(E) = |f(E+\ii 0)|^s (1-\chi_{\tilde{\Delta}}(E))  \ . $$
	By the subharmonicity of  $|f(z)|^s$, we have 
	$$|f(z)|^s \ \le  \ \mathcal{P}g(z) + \mathcal{P}h(z)  $$
	for all $z$ in the upper half plane, $\mathcal{P}g(z)=\frac{1}{\pi}\int g(t)\Im {\frac{1}{t-z}}\dif{t}$ denotes the Poisson integral.    Because the Poisson kernel is a radially decreasing function, 
	$$|f(E+\ii \ve)|^s \ \le \ Mg(E) + Mh(E) \ , $$
	with $Mg$, $Mh$ the Hardy-Littlewood maximal functions of $g,h$, respectively.
	Since $g$ is compactly supported, it follows from \cref{eq:weakL1} that $g\in L^p(\RR)$ for, say, $p=\frac{2}{1+s}>1$, with
		\begin{multline*}\norm{g}_{L_p}^p \ = \ \frac{1}{p} \int_0^\infty t^{p-1} |\{ |g|>t\}| \dif{t} \ \le \ C \int_0^\infty t^{p-1} \min\left (|\tilde{\Delta}|, t^{-1/s}\int_{\RR}\phi \dif{|m|} \right ) \dif{t} \\ \le \ C |\tilde{\Delta}|^{1-sp}  \left ( \int_{\RR} \phi \dif{|m|} \right )^{sp} \ . \end{multline*}
	Because $|f(z)|^s\le \frac{1}{\dist(z,\Delta)^s} \left (\int \phi \dif{|m|} \right)^s$, we have $h\in L^q(\RR)$, for, say, $q=\frac{2}{s}>1$,
	with $$
	    \norm{h}_{L_q}^q \ \le  \ \int_{\tilde{\Delta}^c} \frac{1}{\dist(t,\Delta)^{qs}} \left (\int \phi \dif{|m|} \right)^{qs} \dif{t} \ \le \ C  \left (\int \phi \dif{|m|} \right)^{qs} \ . 
	$$
	Thus
	\begin{align*}\left | \{ E \ : \ \sup_{\ve\in (0,1]} |f(E+\ii \ve)|^s > \alpha\} \right | \ \le& \ \left | \{ E \ : \ Mg(E)\ge \tfrac{\alpha}{2} \} \right | \ + \ \left | \{ E \ : \ Mh(E) \ge \tfrac{\alpha}{2}\} \right | \\
	\le& C \left ( \frac{1}{\alpha^p} \left ( \int_{\RR} \phi \dif{|m|} \right )^{sp}  + \frac{1}{\alpha^q} \left (\int \phi \dif{|m|} \right)^{qs} \right ) \ ,
	\end{align*}
	by the Hardy-Littlewood maximal inequality. 	Using \cref{eq:bound on total variation} and \cref{def:deterministic insulator}, we find that we have shown: 
	\begin{quote} 
	\emph{There are $\nu\in \NN$, $s<1$, and $p,q>1$ such that for $\mu \in \NN$ and $\alpha>0$ we have}
	\begin{multline}\label{eq:xyijweak}
	    \left | \{ E \ : \ \sup_{\ve\in (0,1]} |f^{i,j}_{x,y}(E+\ii \ve)|^s > \alpha\} \right | \\
	    \le \ 
	     C_\mu \left ( \frac{1}{\alpha^p} (1+\norm{x})^{sp\nu} (1+ \norm{x-y})^{-sp\mu}  + \frac{1}{\alpha^q} 
	     (1+\norm{x})^{sq\nu} (1+ \norm{x-y})^{-sq\mu}\right ) \ , 
	\end{multline}
    \emph{for every $x,y\in \ZZ^2$ and $i,j\in \{1,\ldots,N\}$.}
	\end{quote}
	
	To prove \cref{eq:decay of G from SULE} we now apply a Borel-Cantelli argument. Fix  $\nu$, $s<1$ and $p,q>1$ as above and let $\nu'\in \NN$ be such that $\nu'=\nu + 3/\min(sp,sq)$.   Let $\mu'>0$ and apply \eqref{eq:xyijweak} with $\alpha=(1+\norm{x})^{\nu'} (1+ \norm{x-y})^{-\mu'}$ and $\mu> \mu' +3/\min(sp,sq)$ to conclude that
	\begin{multline*}
	    \sum_{x,y,i,j} \left | \left \{ E \ : \ \sup_{\ve\in (0,1]} |f^{i,j}_{x,y}(E+\ii \ve)|^s > (1+\norm{x})^{\nu'} (1+ \norm{x-y})^{-\mu'} \right \} \right | \\
	    \le \ C_{\mu}\sum_{x,y,i,j} (1+ \norm{x})^{-3} (1+\norm{x-y})^{-3} \ < \ \infty \ . 
	\end{multline*}
	We conclude from the Borel-Cantelli lemma that there is a full measure set of energies on which 
	$$ \sup_{\ve\in (0,1]} |f^{i,j}_{x,y}(E+\ii \ve)|^s \ \le \ (1+\norm{x})^{\nu'} (1+ \norm{x-y})^{-\mu'}$$
	for all but finitely many $x,y,i,j$.   Since for each $i,j,x,y$ we have also have $\sup_{\ve\in (0,1]} |f^{i,j}_{x,y}(E+\ii \ve)|^s < \infty$ on a full measure set of $E$, we conclude that there is a full measure set of $E$ on which 
    $$	\sup_{\ve\in (0,1]} |f^{i,j}_{x,y}(E+\ii \ve)|^s \ \le \ C_{\mu'} (1+\norm{x})^{\nu'} (1+ \norm{x-y})^{-\mu'}$$
	Repeating this for each $\mu'\in \NN$ (a countable set) we find that \cref{eq:decay of G from SULE} holds for $E$ in a set of full measure.
	\end{proof}
	
	\section{Proof of \cref{prop:index is zero in mobility gap}}\label{sec:proof of index zero in mob gap} In this section we prove that $\findex \mathbb{Q}U = 0\,$ for $Q$ a projection onto a subset of the mobility gap (this is \cref{prop:index is zero in mobility gap}). Since $Q$ projects onto localized states of $H$, we know it is spanned by a SULE basis $\Set{\psi_n}_n$  as in \cref{def:SULE basis}.
	
	Let us define an operator $V$ on $\im Q$, diagonal in the SULE basis,  
	via $$ V \psi_n := \exp(\ii\arg(x_n\cdot e_1+\ii x_n\cdot e_2)) \psi_n\,. $$ 
	We extend $V$ to $\HH$ by defining $V\psi=\psi$ for $\psi \in \im Q^\perp$.
	Clearly, $V$ is unitary and commutes with $Q$.  Thus $\findex \mathbb{Q} V = 0$, so it suffices to prove $ (U-V)Q=:B$ is compact. We shall show it is Schatten. For this, it suffices to show $$ \sum_y\left(\sum_x|B_{x,x+y}|^p\right)^{1/p}<\infty\,. $$ The proof is similar to that of \cref{lem:flux_piercing}, but here we have the added complication of having to control the infinite collection $\Set{\psi_n}_n$.
	
	Note that $$ B_{x,y} = \sum_{n=1}^\infty (\ee^{\ii\arg(x)}-\ee^{\ii\arg(x_n)}))\psi_n(x)\overline{\psi_n(y)}\,. $$ Now defining $f(x) := \exp(\ii\arg(x))$ we have 
	$$|B_{x,y}|^p  \ \leq \ \left(\sum_{n=1}^\infty |f(x)-f(x_n)||\psi_n(x)||\psi_n(y)|\right)^p \
	\leq \ \sum_{n}|f(x)-f(x_n)|^p|\psi_n(x)||\psi_n(y)| \ ,$$
	where we have used H\"older's inequality in the form $$ (\sum_j a_j b_j c_j)^p \ \leq\  (\sum_j a_j^p b_j c_j) \, (\sum_j b_j^2)^{\frac{p-1}{2}}\, (\sum_j c_j^2)^{\frac{p-1}{2}} $$ as well as the fact that $ \sum_n |\psi_n(x)|^2 \leq 1 $, i.e., $\sum_n \psi_n\otimes\psi_n^\ast = Q \leq \Id$.
	
	Now we observe that $f$ obeys the estimate in \cref{eq:admissible flux functions}, so that using \cref{eq:SULE cond} we find \begin{align*} |B_{x,x+y}|^p \ &\leq \ D^pC_\mu^2\sum_{n} \frac{\norm{x-x_n}^p}{(1+\norm{x_n})^{p/2}(1+\norm{x})^{p/2}}(1+\norm{x-x_n})^{-\mu}(1+\norm{x+y-x_n})^{-\mu+p}(1+\norm{x_n})^{2\nu}\\
		&\leq\  D^pC_\mu^2\sum_{n} (1+\norm{x-x_n})^{-\mu+p}(1+\norm{x+y-x_n})^{-\mu+p}(1+\norm{x_n})^{2\nu-p/2}(1+\norm{x})^{-p/2}\\
		&\leq\  D^pC_\mu^2(1+\norm{y})^{-\mu/2+p/2}(1+\norm{x})^{-p/2}\sum_{n} (1+\norm{x_n})^{2\nu-p}
	\end{align*}
	where in the last step we have used the triangle inequality in the form $ (1+\norm{a})(1+\norm{a+b}) \geq 1+\norm{b} $ as well as $(1+\norm{x+y-x_n})^{-\mu/2+p/2}\leq 1$, $(1+\norm{x-x_n})^{-\mu/2+p/2}\leq 1$.
	
	The result now follows thanks to \cref{eq:SULE summability} and the fact $p,\mu$ may be chosen arbitrarily large.

	\bigskip
	
	\noindent\textbf{Acknowledgements:}  
	We wish to thank Gian Michele Graf and Martin Zirnbauer for useful discussions. Work on this project was supported in parts by the following grants:  
	A. Bols was supported by the Villium Fonden through the QMATH Centre of Excellence, grant no. 10059.
	J. Schenker was supported by the U.S. National Science Foundation under Grant No. (1900015).
	J. Shapiro  was supported by the Swiss National Science Foundation (grant number P2EZP2\_184228), and the Princeton-Geneva Univ.  collaborative travel funds.
	
	\emph{Data availability and conflict of interest} This manuscript has no associated data available and no known conflicts of interest.
	\bigskip

	\begingroup
	\let\itshape\upshape
	\printbibliography
	\endgroup
\end{document}